\documentclass[a4paper,twocolumn,11pt, accepted = 2022-06-13]{quantumarticle}
\pdfoutput=1
\usepackage[utf8]{inputenc}
\usepackage[english]{babel}
\usepackage[T1]{fontenc}
\usepackage{hyperref}
\usepackage[numbers,sort&compress]{natbib}

\usepackage{graphicx}
\usepackage{multirow}
\usepackage{amsmath}
\usepackage{mathtools}
\usepackage{enumitem}
\usepackage{amssymb}
\usepackage{amsthm}
\usepackage{soul}
\usepackage{verbatim}
\usepackage{bm}
\usepackage[table]{xcolor}
\usepackage{tikz}
\usepackage{xspace}
\usepackage{import}
\usepackage[caption=false]{subfig}
\usepackage{bbold}

\makeatletter
\newcommand\footnoteref[1]{\protected@xdef\@thefnmark{\ref{#1}}\@footnotemark}
\makeatother

\newcommand{\ket}[1]{\ensuremath{\left\lvert{#1}\right\rangle}}

\DeclareMathOperator*{\argmax}{arg\,max}
\DeclareMathOperator*{\argmin}{arg\,min}

\newcommand{\Vv}{\mathcal{V}}

\newcommand{\CC}{\mathbb{C}}
\newcommand{\NN}{\mathbb{N}}
\newcommand{\RR}{\mathbb{R}}
\newcommand{\ZZ}{\mathbb{Z}}

\newtheorem{dfn}{Definition}[section]
\newtheorem{lem}[dfn]{Lemma}

\newtheorem{ntn}[dfn]{Definition}
\newtheorem{thm}[dfn]{Theorem}
\newtheorem{cor}{Corollary}[dfn]
\newtheorem{alg}[dfn]{Algorithm}

\begin{document}

\title{Heisenberg-limited quantum phase estimation of multiple eigenvalues with few control qubits}

\author{A. Dutkiewicz}
\affiliation{Instituut-Lorentz, Universiteit Leiden, 2300 RA Leiden, The Netherlands}
\author{B.M. Terhal}
 \affiliation{QuTech, Delft University of Technology, P.O. Box 5046, 2600 GA Delft, The Netherlands and JARA Institute for Quantum Information, Forschungszentrum Juelich, D-52425 Juelich, Germany}
\author{T. O'Brien}
\affiliation{Google Quantum AI, 80636 Munich, Germany}
\affiliation{Instituut-Lorentz, Universiteit Leiden, 2300 RA Leiden, The Netherlands}

\begin{abstract}
    Quantum phase estimation is a cornerstone in quantum algorithm design, allowing for the inference of eigenvalues of exponentially-large sparse matrices.
    The maximum rate at which these eigenvalues may be learned, --known as the Heisenberg limit--, is constrained by bounds on the circuit complexity required to simulate an arbitrary Hamiltonian.
    Single-control qubit variants of quantum phase estimation that do not require coherence between experiments have garnered interest in recent years due to lower circuit depth and minimal qubit overhead.
    In this work we show that these methods can achieve the Heisenberg limit, {\em also} when one is unable to prepare eigenstates of the system. Given a quantum subroutine which provides samples of a `phase function' $g(k)=\sum_j A_j e^{i \phi_j k}$ with unknown eigenphases $\phi_j$ and overlaps $A_j$ at quantum cost $O(k)$, we show how to estimate the phases $\{\phi_j\}$ with (root-mean-square) error $\delta$ for total quantum cost $T=O(\delta^{-1})$. Our scheme combines the idea of Heisenberg-limited multi-order quantum phase estimation for a single eigenvalue phase \cite{Higgins09Demonstrating,Kimmel15Robust} with subroutines with so-called dense quantum phase estimation which uses classical processing via time-series analysis for the QEEP problem \cite{Somma19Quantum} or the matrix pencil method. For our algorithm which adaptively fixes the choice for $k$ in $g(k)$ we prove Heisenberg-limited scaling when we use the time-series/QEEP subroutine. We present numerical evidence that using the matrix pencil technique the algorithm can achieve Heisenberg-limited scaling as well.
\end{abstract}

\maketitle

\tableofcontents

\section{Introduction}
For quantum computers to overcome the $50^+$ year head start in research and development enjoyed by their classical competition, quantum algorithms must eke out every inch of quantum speedup over their classical counterparts. Quantum phase estimation (QPE) of a unitary operator $U$, --a BQP-complete problem~\cite{Wocjan06Several}--, plays a central or support role in many promising quantum applications~\cite{Shor95Polynomial,Harrow09Quantum,Whitfield10Simulation}.

However, not all flavours of quantum phase estimation are equally powerful.
To perform QPE with accuracy (root-mean-square) error $\delta$, initial implementations of quantum phase estimation~\cite{Nielsen00Quantum, cleve+} used a $O(\log(\delta^{-1}))$-qubit control register, and required computation time scaling as $T=O(\delta^{-2})$, known as the Sampling Noise Limit, when contributions from outlying (unlikely) data are taken into account~\cite{Higgins09Demonstrating}.
Much work has been undertaken over the succeeding years to improve QPE estimation rates to the theoretically optimal Heisenberg limit $T=O(\delta^{-1})$~\cite{giovannetti2006quantum, vandam2007optimal,berry2009perform} and reduce the control overhead.
 

The requirement to perform applications of $U$ conditional on a large entangled control register is technically challenging and has strong coherence requirements.
It has been known for a long time \cite{griffiths} that the control register in QPE can be replaced by a single qubit using classical feedback and re-preparation of the control qubit, also known as iterative QPE~\cite{Kitaev95Quantum}.
For estimating the phase of a single eigenstate, --assuming the preparation of this eigenstate--, the sampling noise limit can thus simply be achieved using a single control qubit. 
In \cite{Higgins09Demonstrating} iterative-QPE was extended to achieve the Heisenberg limit $T=O(\delta^{-1})$. This Heisenberg limit can be shown, via Cramer-Rao bounds~\cite{Higgins09Demonstrating}, to be a lower bound on the cost of phase estimation, assuming one cannot fast-forward the unitary $U$ \cite{Berry07Efficient}. This type of estimation have additionally demonstrated a relative robustness to error~\cite{Kimmel15Robust,Wiebe16Efficient}, which is of interest to NISQ applications. Other analyses of QPE use maximum-likelihood~\cite{Svore13Faster}, or Bayesian~\cite{Wiebe16Efficient,Berg20Efficient} inference. 

The requirement to prepare eigenstates of the unitary $U$ is not possible for most applications.
It is well known that the `textbook' QPE algorithm succeeds for any initial state, i.e. the output is always an accurate estimate of one of the eigenphases of $U$~\cite{Nielsen00Quantum}.
However, the performance of few-ancilla QPE on starting states that are not eigenstates has only been examined recently.
In~\cite{Obrien19Quantum} it was demonstrated numerically that one may infer single eigenvalues from mixed or superposed initial eigenstates using standard classical signal processing techniques~\cite{Rife74Single}.
Under some additional constraints on the system, it was recently found that these techniques could be performed in the absence of any control qubits or the need to apply controlled unitary operations~\cite{Lu20Algorithms,Obrien20Error}, a further significant saving. Due to the need to `densely sample' the phase function $g(k)=\sum_j A_j e^{i \phi_j k}$, and a lack of optimization of the classical post-processing techniques, Ref.~\cite{Obrien19Quantum} only achieved sampling-noise-limited scaling, but not Heisenberg-limited scaling. By dense sampling we mean that we draw samples from $g(k)$ with $k$ a sequence of integers, $k=0, 1, \ldots, K$ (as opposed to, say, choosing $k=2^d$, i.e. exponentially increasing which is used in textbook QPE and iterative QPE). By a clever adjustment of the quantum phase estimation problem to target estimation of the spectral function, Eq.~\eqref{eq:spectral_function}, of the input state, Ref.~\cite{Somma19Quantum} was able to prove rigorous results, with bounds that were subsequently improved in Ref.~\cite{Roggero20Spectral}. Still, these results fall short of reaching ``the Heisenberg limit for the problem of estimating multiple phases", however that should be defined.

In this work, we demonstrate single-control qubit quantum phase estimation at a so-called Heisenberg limit. We extend the methods used in Refs.~\cite{Higgins09Demonstrating,Kimmel15Robust} that obtain Heisenberg-limited scaling for single eigenphases to the multiple-phase setting by the use of a multi-order scheme and phase matching subroutines between different orders. We show that to make this phase matching unambiguous requires the sampling scheme to be adaptive, i.e. the next choice for $k$ of $g(k)$ depends on the current phase estimates. At each order the multi-order algorithm requires input from a dense phase estimation method: for a given order $k$ we use samples from $g(k{\sf k})$ with ${\sf k}=0, 1,\ldots, K$.
As we require the freedom to choose $k$ a real number, to be applicable to a completely general $U$ we must invoke the quantum singular value transformation of Ref.~\cite{Gilyen19QSVT}, which requires $O(1)$ additional control bits.
Using the time-series or QEEP analysis of Ref.~\cite{Somma19Quantum} as classical processing subroutine, we are able to obtain a rigorous proof of Heisenberg-limited scaling of our multi-order scheme.
Using the matrix pencil method analysed in Ref.~\cite{Obrien19Quantum}, as such a dense subroutine, we are able to show numerical results consistent with the Heisenberg limit, with a performance improvement over the time-series analysis results.

In essence, our paper is concerned with what choices of $k$ in $g(k)$ and what classical processing are needed to enable Heisenberg-limited scaling, i.e. scaling which minimizes the total number of applications $T$ of (controlled) $U$ (which we refer to as the quantum cost) given a targeted error $\delta$ with which to estimate multiple eigenvalue phases of $U$ present in some input state $\ket{\Psi}$. It can thus be viewed as purely solving a problem of classical signal processing. This does not mean that such questions are trivial: for example, the question of how to estimate phases if one is allowed to only get single samples from $g(k)$ for a set of randomly chosen $k$ relates to the dihedral hidden subgroup problem in quantum information theory \cite{regev}. 
We note that other work, based on a Monte Carlo extension of \cite{Somma19Quantum}, achieving Heisenberg scaling (up to polylog factors) was recently presented in \cite{LT:heisenberg}.
Another recent work also demonstrated numerical evidence for Heisenberg-limited phase estimation using Bayesian methods~\cite{Gebhart20Heisenberg}.
We also note that the information-theoretically optimal method in \cite{Svore13Faster} which picks random $k$ in $g(k)$ has a classical processing cost which is linear in the quantum cost $T$, and Theorem 1 in \cite{Svore13Faster} can be converted to bound the mean-squared-error in estimating a single phase, see comments below Theorem \ref{thm:heisenberg_limit} in Section \ref{sec:speed_limits}. 
In principle this information-theoretic method can be extended to the case of $n_{\phi}$ eigenvalue phases, but the classical processing cost will be exponential in $n_{\phi}$ as one iterates over the possible values of the phases, while our Algorithm \ref{alg:adaptive} has a polynomial (but superlinear) quantum and classical cost in terms of the number of phases $n_{\phi}$.

\subsection{Outline}

We begin in Sec.~\ref{sec:QPE_split} by defining a few mathematical objects. We separate the quantum part of a phase estimation problem, namely sampling of the phase (or signal) function $g(k)$ in Eq.~\eqref{eq:phase_function} given an input unitary $U$ and input state $|\Psi\rangle$ in Definition \ref{def:phase_function_estimation_problem} through running some quantum circuits, and the classical processing of samples from $g(k)$ to extract the eigenvalue data of $U$. In Sec.~\ref{sec:distcirc} we state and prove some needed properties of the distance between phases. In Sec.~\ref{sec:speed_limits} we prove several Cramer-Rao bounds on the scaling of the error versus the total quantum cost for the estimation of a single eigenvalue phase, Theorem \ref{thm:heisenberg_limit}. We state the previous result on getting Heisenberg-limited scaling for a single eigenvalue phase in Algorithm \ref{alg:single_phase_Heisenberg_limit}.
Table \ref{tab:glossary} contains a glossary of the symbols used in this paper.

In Sec.~\ref{sec:QPE_sparse_def} we properly define a multi-eigenvalue phase estimation problem (Def.~\ref{def:MEEP}).
We discuss algorithms that can be used to extract multiple phases from densely-sampled signal $g(k)$.
We state error bounds satisfied by the output of Alg.~\ref{alg:phase_extraction} (Lemma \ref{lem:phase_extraction_promises}) that is used as the fixed-order subroutine in our final Algorithm \ref{alg:adaptive} which achieves Heisenberg scaling.

In Section \ref{sec:phase_matching_problem} we present our Heisenberg-limited algorithm. We discuss a critical aliasing problem to be solved which occurs when estimating multiple eigenvalues. We show that an adaptive choice for $k$ in $g(k)$ can solve this issue and we prove that such adaptive choice always exists in Lemma~\ref{lem:phase_matching_solution}. In Lemma \ref{lem:phase_matching_checks} and Lemma \ref{lem:phase_matching_with_failures} we prove some properties about Algorithm \ref{alg:adaptive} which will be helpful in proving the final Theorem \ref{thm:adaptive_gets_Heisenberg_limit}.

Thus in Theorem~\ref{thm:adaptive_gets_Heisenberg_limit} we prove Algorithm~\ref{alg:adaptive} achieves Heisenberg-limited scaling, which is the main result of this work. In Sec.~\ref{sec:numerics} we numerically compare this rigorous implementation to an implementation using the matrix pencil method, used in Ref.~\cite{Obrien19Quantum}, for which we are unable to find a rigorous proof of Heisenberg scaling. We finish the paper with a discussion, Section \ref{sec:conclusion}.

\section{The classical and quantum tasks of phase estimation}\label{sec:QPE_split}

\begin{table*}[]
    \centering
    \begin{tabular}{p{0.07\linewidth}|p{0.24\linewidth}|p{0.65\linewidth}}
         Symbol & Term & Description\\
         \hline
         $U$ & Unitary & The unitary whose eigenphases we wish to estimate.\\
         $T$ & Quantum cost & The total number of applications of controlled $U$ over the course of the phase estimation algorithm.\\
         $\phi_j$ & Phase & A number $\phi_j \in [0, 2\pi)$ such that the $j$th eigenvalue of $U$ is $e^{i\phi_j}$.\\
         $A_j$ & Overlap & The overlap of the input state and the $j$th eigenstate of $U$; see Def.~\ref{def:phase-function}.\\
         $g(k)$ & Signal / Phase function & See Def.~\ref{def:phase-function}.\\
         $d$ & Order & Running index for the order of estimation. At each order we construct new phase estimates $\tilde\phi_j^{(d)}$ using new data and the previous estimates $\tilde\phi_j^{(d-1)}$.\\
         $\tilde \phi_j^{(d)}$ & Estimate & Estimate of $\phi_j$ obtained at the $d$th order; see step~\ref{step:update} of Alg.~\ref{alg:adaptive}.\\
         $k_d$ & Exponent & At order $d$ we perform the QEEP subroutine for $V = U^{k_d}$.\\
         $\kappa_d$ & Multiplier & Defines $k_d = \kappa_d k_{d-1}$.\\
          $\theta_j^{(d)}$ & Phase & Eigenphase of $U^{k_d}$.\\
         $\tilde\theta_j^{(d)}$ & Estimate & Estimate of $\theta_j^{(d)}$; see step~\ref{step:exeQEEP} of Alg.~\ref{alg:adaptive}.  \\
         $\epsilon$ & Single order error & Error parameter used for the QEEP subroutine at each order; see Alg.~\ref{alg:phase_extraction}.\\
         $\delta$ & Final error & A bound for standard deviation of the final estimates; see Def.~\ref{def:MEEP}.\\
         $p_d$ & Confidence bound & The probability with which the QEEP subroutine at order $d$ succeeds; see Eq.~\eqref{eq:defpd}.\\
         $A$ & Overlap bound & We wish to estimate the phases for which the overlap is $A_j > A$; see Def.~\ref{def:MEEP}. \\
    \end{tabular}
    \caption{Glossary of symbols used in the main text.}
    \label{tab:glossary}
\end{table*}

One may separate quantum phase estimation into the extraction of a signal which consists of oscillations at eigenvalue frequencies $\phi_j$ at a chosen time $k$, and the processing of this signal to resolve the frequencies. Let us first define the following:

\begin{ntn}[Signal or Phase Function, and Spectral Function] \label{def:phase-function}
Let $U\in \mathsf{U}(2^N)$ be an $N$-qubit unitary operator, and $|\Psi\rangle\in\CC^{2^N}$ an $N$-qubit state. We label the eigenstates $|\phi_j\rangle$ of $U$ by their phase --- $U|\phi_j\rangle=e^{i\phi_j}|\phi_j\rangle$. We can decompose $|\Psi\rangle$ in terms of these eigenstates,
\begin{equation}
    |\Psi\rangle = \sum_j a_j|\phi_j\rangle,
\end{equation}
and write the overlap $A_j:=|a_j|^2, \sum_j A_j=1$.
We define the phase function, -- also called the signal--, $g(k)$ for $k \in \mathbb{R}$ of a state $|\Psi\rangle$ under $U$ as
\begin{equation}
    g(k)=\sum_jA_je^{ik\phi_j}.\label{eq:phase_function}
\end{equation}
 The spectral function $A(\phi)$ is defined as
    \begin{equation}
        A(\phi)=\sum_jA_j\delta(\phi-\phi_j).\label{eq:spectral_function}
    \end{equation}
    Note that $\int_{0}^{2\pi}d\phi A(\phi)=1$, and $g(k)=\int_0^{2\pi}d\phi e^{ik\phi}A(\phi)$; i.e. the phase function sets the Fourier coefficients of the spectral function. 
\end{ntn}

Note that one may change seamlessly between the description of a unitary $U$ and its eigenvalues and a Hermitian operator $H$ and its eigenvalues using the transform $U=e^{iHt}$ for an appropriate choice of $t$. Note that since $g(-k)=g^*(k)$ we can restrict ourselves to $k\geq 0$.

One may consider algorithms estimating $g(k)$ at integer $k \in \mathbb{Z}^+$, which require the quantum circuits using controlled-$U^k$ in Fig.~\ref{fig:QPE} with $k\in \mathbb{Z}^+$.
In our final Alg.~\ref{alg:adaptive} we will however use $k\in \mathbb{R}^+$ (in practice $k \in \mathbb{Q}^+$).
In order to implement $U^k$, we can write $k=\lfloor k\rfloor + \alpha$, and we can simulate $U^k$ in time $O(k)$ if we can simulate $U^\alpha$ in time independent of $k$.
The accuracy of this fractional query to $U^{\alpha}$ can be independent of the final error in our phase estimation (as long as it is sufficiently small).
Simulating $U^{\alpha}$ is not a significant issue for Hamiltonian simulation methods such as Trotter decompositions~\cite{Childs20Theory}, which allow simulation of $e^{iHt}$ for arbitrary $t\in\RR$.
If we instead have access to a circuit implementation of a unitary $U$, or a block-encoding of a Hamiltonian $H$, we can implement a fractional query of $U^{\alpha}$ via the quantum singular value transform~\cite[Corollary 34]{Gilyen19QSVT}.
The circuit to implement $U^{\alpha}$ to error $\epsilon$ requires ${\cal O}(1)$ additional ancilla qubits, and ${\cal O}(\Delta^{-1}_{\max}\log(1/\epsilon))$ implementations of controlled-$U$ (where $\Delta_{\max}$ is the largest gap in the spectrum of $U$)~\footnote{The $\Delta_{\max}$ dependence in our circuit comes from the requirement in~\cite[Corollary 34]{Gilyen19QSVT} that our unitary have spectrum on $[-\pi +\Delta_{\max}, \pi - \Delta_{\max}]$. Though this will not immediately be the case, we are free to rotate the spectrum of our unitary to satisfy this requirement, as long as the spectral gap exists. We also only require to consider those eigenstates with support on our initial state in this algorithm; eigenstates with zero weight can be adjusted as part of implementing the quantum singular value transformation without affecting the phase function $g(k)$.}.
We will assume in this work that our states have support on at most $n_{\phi}$ phases, and so we can bound $\Delta\geq \pi/n_{\phi}$.
The cost of implementing $U^{\alpha}$ will thus not be a significant part of the cost to implement $U^k$ under the assumption $k>> A_0^{-1}\log(1/\epsilon)$.
To simplify our remaining analysis, we assume herein that we can implement $U^k$ at a total quantum cost $k$ for all $k\in\mathbb{R}_+$.



The following task summarizes the quantum subroutine for phase estimation which is to be executed with the quantum circuits in Fig.~\ref{fig:QPE}:

\begin{dfn}[Phase Function Estimation, PFE] \label{def:phase_function_estimation_problem}
Let $U$ be an $N$-qubit unitary operator and $|\Psi\rangle$ an $N$-qubit quantum state. Assume
\begin{enumerate}
    \item A quantum circuit implementation of $U$ (conditional on a control qubit) and,
    \item A quantum circuit that prepares $|\Psi\rangle$.
    \end{enumerate}
    Given a $k\in\ZZ^+$, error $\epsilon>0$, and confidence $0<p\leq 1\in\RR$, {\em PFE} outputs an estimate $\tilde{g}(k)$ of the phase function $g(k)$ of $|\Psi\rangle$ under $U$, with $\mathbb{P}(|\tilde{g}(k)-g(k)|\leq \epsilon) \geq 1-p$ with quantum cost $T= M |k|$ where $M$ is the number of repetitions of both experiments in Fig.~\ref{fig:QPE} and $M=\Theta(|\ln(1-p)| \epsilon^{-2})$ via a Chernoff bound. Our assumption\footnote{In practice, the cost of phase function estimation for $k\in\mathbb{R}^+$ will scale as $M(\lfloor k\rfloor + O(n_{\phi}\log(1/\epsilon))$ when using quantum signal processing techniques. One can confirm that the effect this has on our result is to effectively increase the cost of calling the QEEP subroutine, Alg.~\ref{alg:phase_extraction}, from $O(\epsilon^{-6})$ to $O(\epsilon^{-6}\log(1/\epsilon))$. This will only change the prefactor of our Heisenberg-limited phase estimation algorithm (Alg.~\ref{alg:adaptive}), and does not prevent it achieving the Heisenberg limit.} that the cost of implementing $U^k$ for $k\in \mathbb{R}^+$ scales as $O(k)$ implies that the above scaling for the cost of phase function estimation holds when $k \in \mathbb{R}^+$.
\end{dfn}

Note that estimating the quantum cost of the subroutine in Def.~\ref{def:phase_function_estimation_problem} as linear in $k$ is consistent with general no-fast forwarding statements \cite{Berry07Efficient} which state that for general Hamiltonians one cannot implement $U^t=\exp(i t H)$ in time sub-linear in $t$. It is expected that phase function estimation is hard to do efficiently on a classical computer as it allows one, via classical post-processing, to sample from the eigenvalue distribution from the input state which can be reformulated as a BQP-complete problem~\cite{Wocjan06Several}.

\begin{figure}[h!]
\centering
\includegraphics[width = 0.5\textwidth]{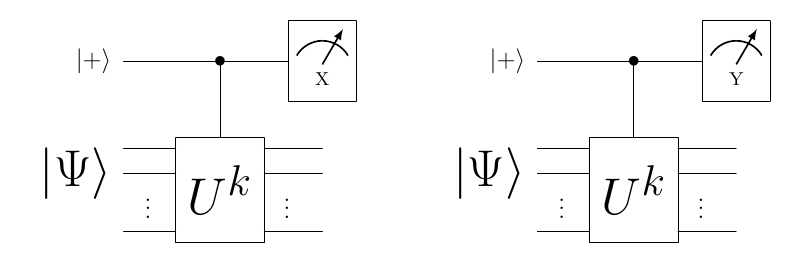}
\caption{For an input state $\ket{\Psi}=\sum_j a_j \ket{\phi_j}$ and initial ancilla state $\ket{+}=\frac{1}{\sqrt{2}}(\ket{0}+\ket{1})$ the probability for the ancilla measurement outcome to be $\pm 1$ is $\mathbb{P}(\pm 1)=\frac{1}{2}\sum_j A_j [1\pm \cos(k\phi_j)]$ (left circuit measuring in the $X$-basis) and $\mathbb{P}(\pm 1)=\frac{1}{2}\sum_j A_j [1\mp \sin(k\phi_j)]$ (right circuit measuring in the $Y$-basis).}
\label{fig:QPE}
\end{figure}

The quantum subroutine for PFE proceeds by executing the circuits in Fig.~\ref{fig:QPE} for the given $k$. The control qubit is prepared in the $\frac{1}{\sqrt{2}}(|0\rangle+|1\rangle)$ state, and is used to control $k$ applications of the unitary $U$ on the system register prepared in $|\Psi\rangle$. The reduced density matrix of the control qubit then takes the form
\begin{equation}
    \rho = \frac{1}{2}\left(\begin{array}{cc}1 & g(k) \\ g^*(k) & 1\end{array}\right).\label{eq:control_rdm}
\end{equation}
The phase function $g(k)$ is extracted by state tomography of the control qubit: one estimates the real and imaginary parts of $g(k)=g^{(r)}(k)+ig^{(i)}(k)$ by $M$ repetitions of the circuit in Fig.~\ref{fig:QPE} and measurements of the control qubit in the X- or Y-basis respectively. 
We will ignore any dependence of phase function estimation on $N$.

\subsection{Phase distance on the circle}
\label{sec:distcirc}

Quantum phase estimation describes a series of protocols to estimate the eigenphases $\phi_j$. As these eigenvalues are defined on the circle $[0,2\pi)$, we need a notion of distance which respects this periodicity:
\begin{dfn}
For $x\in\RR$ we define the distance $|\cdot |_{T}\in [0,\pi]$ as
\begin{align}
    |x|_{T}:=\min_{m\in \mathbb{Z}} (|\Delta|),
    &\mbox{ with } \; x=\Delta+2 \pi m,\nonumber\\
    &\mbox{ with } \Delta \in [-\pi, \pi).
\end{align}
Clearly, the distance obeys the triangle inequality: for $x_1, x_2 \in \RR$
\begin{equation}
    |x_1 + x_2|_T \leq |x_1|_T + |x_2|_T.\label{eq:triangle_inequality}
\end{equation}
\end{dfn}

The following Lemma addresses a technical issue in the proof of performance of Algorithm \ref{alg:adaptive}. For integer $k \in \mathbb{Z}^+$, we have $|k x|_T=\min_{m \in \mathbb{Z}} |k \Delta+2 \pi m|=k\min_{m \in \mathbb{Z}} |\Delta+\frac{2 \pi m}{k}|_T$, implying Eq.~\eqref{eq:normscaling} below directly. However, for $k \in \mathbb{R}^+$ (or rational numbers $k \in \mathbb{Q}^+$) we need to specify a range of $x$ for which such a statement holds, that is: 
\begin{lem}\label{lem:normscaling}
Suppose $k > 1$ and $\theta,\phi\in[0,2\pi)$. If \begin{equation}
    \frac{\pi}{k}\leq \phi \leq \frac{\pi(2\lfloor k\rfloor-1)}{k},\label{eq:normscaling_condition}
\end{equation}
we have for any $\theta$
{\small
\begin{equation}
    \min_{n\in\{0,\ldots,\lfloor k\rfloor-1\}}\left|\phi-\frac{\theta}{k}-\frac{2\pi n}{k}\right|_T = \frac{1}{k}\left|k\phi - \theta\right|_T.\label{eq:normscaling}
\end{equation}}

\end{lem}
\begin{proof}
Let $x=\phi-\frac{\theta}{k}$, then Eq.~\eqref{eq:normscaling_condition} and $\theta \in [0,2\pi)$ imply that
\begin{equation}
    -\pi \leq k x \leq 2\pi \lfloor k \rfloor-\pi.  
\end{equation}
and thus $kx=\Delta+2 \pi m$ with $m \in \{0,\ldots \lfloor k \rfloor -1\}$ and $\Delta \in [-\pi, \pi)$ and $|k \phi-\theta|_T=|\Delta|$. Hence $x=\frac{\Delta}{k}+\frac{2 \pi m}{k}$ with $m \in \{0,\ldots \lfloor k \rfloor -1\}$, implying Eq.~\eqref{eq:normscaling} where the minimum can be achieved by $m=n$.
\end{proof}

\subsection{Limits for single-eigenvalue phase estimation}\label{sec:speed_limits}

For the special case of estimating a single eigenvalue phase, the Cram\'er-Rao theorem can be used to lower bound the quantum cost $T$ to learn the phase with accuracy $\delta$, known as the Heisenberg limit. The problem of estimating multiple phases $\phi_j$, in the presence of unknown overlaps $A_j$, is not easily amenable to such Fisher information analysis as it is a multi-parameter estimation problem. However, it can be expected that the cost $T$ of this task is at least as high as that of single phase estimation, hence it is of interest to review these bounds here. The following theorem also proves a dense signal limit which sits in between Heisenberg and sampling noise scaling:

\begin{thm}[The Heisenberg, Dense Signal and Sampling Limits]\label{thm:heisenberg_limit}
The (root-mean-square) error $\delta$ of an estimator $\tilde{\phi}$ of the eigenvalue phase $\phi$ employing the circuits in Fig.~\ref{fig:QPE} on a eigenstate of U is always lower bounded as 
\begin{equation}
  \mbox{ {\rm Heisenberg Limit}}\colon  \delta \geq c T^{-1}, 
\end{equation}
where $T$ is the quantum cost of implementing the circuits. If we choose to use only quantum circuits with $k=1$, the sampling noise limit holds:
\begin{equation}
   \mbox{ {\rm Sampling Noise Limit}}\colon \delta  \geq T^{-1/2}. 
\end{equation}
If we choose circuits with $k=0,1, \ldots, K$ with a fixed number of repetitions $M$ for each circuit we are bound by a so-called `dense signal' limit:
\begin{equation}
  \mbox{ {\rm Dense Signal Limit}}\colon  \delta  \geq c T^{-3/4}. 
  \label{eq:DSL}
\end{equation}
In these statements $c$ is some constant.
\end{thm}
\begin{proof} Let $\tilde{\phi}$ be an estimator of $\phi$ which is inferred from the data ${\bf x}$. Here the data ${\bf x}$ is the string of outcomes of the ancilla qubit measurements for all the experiments using the left and right circuits in Fig.~\ref{fig:QPE}. We have 
\begin{equation}
    \delta^2=\sum_{\bf x}\mathbb{P}({\bf x}|\phi) (\phi-\tilde{\phi}({\bf x}))^2 \geq I^{-1}(\phi),
    \end{equation}
    by the Cram\'er-Rao theorem \cite{Cramer46Mathematical,Rao45Information}
    where the Fisher information is defined as 
\begin{equation}
    I(\phi)=\sum_{{\bf x}} \mathbb{P}({\bf x}|\phi)\left\{\frac{\partial}{\partial \phi}\ln\left[\mathbb{P}({\bf x}|\phi)\right]\right\}^2.\label{eq:fisher_def}
\end{equation}
Thus $I(\phi)$ limits the information we may learn about $\phi$ given a dataset ${\bf x}$ drawn from $\mathbb{P}({\bf x}|\phi)$ and we can calculate $I(\phi)$.  Let $M_k^r$ be the number of experiments, using the circuit with the $X$ measurement, and $M_k^i$ be the number of experiments using the circuit with the $Y$ measurement with a certain chosen $k$.  
 The Fisher information for all independent experiments together is additive, i.e. $I(\phi)=\sum_k [M_k^r I(\phi|k,r)+M_k^i I(\phi|k,i)]$ with $I(\phi|k,r)$ and $I(\phi|k,i)$ the Fisher information of a single experiment and $\sum_k$ is the sum over the chosen set of $k$s. For a single experiment we can calculate, using Eq.~\eqref{eq:fisher_def} and the probability for the output bit given in Fig.~\ref{fig:QPE}, that $I(\phi|k,r)=I(\phi|k,i)=k^2$ and thus
 \begin{equation}
   I(\phi)=\sum_{k} k^2 (M_k^r +M_k^i).
   \label{eq:Fisher_info_single_eigenstate}
 \end{equation}
At the same time the total quantum cost of all experiments is
\begin{equation}
    T=\sum_{k} k(M_k^r+M_k^i).\label{eq:Total_quantum_cost_single_eigenstate}
\end{equation}
The key insight here is that the relative dependence on $k$ is different between $T$ and $I$ and this implies that the trend of the number of experimental runs $M_k^{r/i}$ as a function of $k$ will affect the maximum rate of estimation. 
If we choose only $k=1$ we see that $\delta^2 \geq \frac{1}{T}$ which is the sampling noise limit.

The biggest value for $I(\phi)$ for a given $T$ is obtained when we choose a single largest possible $k=K$ so that $T=K M_K$ and $I(\phi)=K^2 M_K=T^2/M_K$ with $M_K=M_K^r+M_K^i$ . This implies a Heisenberg limit, i.e. $\delta \geq c T^{-1}$ where $c$ is some constant depending on $M_K$.
If we however make the `dense signal' choice, that is, $M^r_k=M^i_k=M$ for $k=1,2 \ldots K$, then 
\begin{equation}
    I(\phi)=\frac{M}{3}K(K+1)(1+2K),
\end{equation}
while the total quantum cost is $T=M K(K+1)$.
Increasing $K$ with $M$ fixed, we have to leading order in $K$ that $\delta \geq I(\phi)^{-\frac{1}{2}}=\sqrt{\frac{3}{2}}M^{-\frac{1}{2}}K^{-\frac{3}{2}}=\sqrt{\frac{3}{2}} M^{\frac{1}{4}}T^{-\frac{3}{4}}$, which is the dense signal limit.

\end{proof}

{\em Remark}: We note that a randomized version of the dense signal choice can potentially scale in near-Heisenberg-limited fashion. In this method, one would draw $k$ at random from $1, \ldots, K$ and repeat this $S$ times to generate random variables $k_1, \ldots, k_S$ and repeat experiments with fixed $M$ for each such $k_i$. With the right choice of $S \times M={\rm polylog}(K)$, one can argue, using the Cramer-Rao lower bound analysis above, that the expected error $\mathbb{E}(\delta)\geq \frac{{\rm poly log}(\mathbb{E}(T))}{\mathbb{E}(T)}$ where $\mathbb{E}(T)$ is the expected quantum cost. The algorithm in \cite{Svore13Faster} uses such strategy with $M=1$.
Clearly, the sampling noise limit can be achieved by choosing $k=1$ in the circuits of Fig.~\ref{fig:QPE}. However, one can ask whether the dense signal limit or the Heisenberg limit can also be achieved, in particular when we demand that the classical post-processing is computationally efficient, meaning that this processing is polynomial in the quantum cost $T$.
For the dense signal limit one needs a classical method to process the estimates of $g(k)$ at $k=1,\ldots, K$ to estimate $\tilde{\phi}$. Using perturbation theory in the noise, the matrix pencil method has been claimed to achieve this for a single eigenvalue \cite{Hua90Matrix}.

Achieving the Heisenberg limit is non-trivial due to phase aliasing: the phase function $g(k)$ obtained by the experiments using $U^k$ remains invariant if the phase $\phi$ is shifted by $\frac{2\pi}{k}$.
This implies that a strategy of estimating $\phi$ from a single point $g(k)$ at large $k$ will fail unless $\phi$ is already known to sit within a window of width $\frac{2\pi}{k}$. This issue is circumvented by sampling $g(k)$ at multiple orders $k=2^d$ to `gradually zoom in' on $\phi$. To get Heisenberg scaling, one lets the number of samples $M$ and thus the confidence, to depend on the order, so that the most significant bits of $\phi$ are determined with the highest confidence. Methods for doing this were first introduced in Ref.~\cite{Higgins09Demonstrating}, and improved in Ref.~\cite{Kimmel15Robust} for the purpose of gate calibration. Here we state the result:
\begin{alg}[Heisenberg Algorithm For Single Eigenvalue Phase \cite{Higgins09Demonstrating,Kimmel15Robust}]\label{alg:single_phase_Heisenberg_limit}
Given an targeted error $\delta> 0$, and numbers $\alpha, \gamma \in \mathbb{Z}^+$.
The Heisenberg algorithm which outputs an estimate $\tilde{\phi}$ for $\phi$ proceeds as follows:
\begin{enumerate}
    \item Fix $d_f = \lceil \log_2(1/\delta)\rceil$.
    \item For $d=0,1\ldots, d_f$:
    \begin{enumerate}
        \item Use the {\rm PFE} subroutine, Def.~\ref{def:phase_function_estimation_problem}, circuits in Fig.~\ref{fig:QPE}, to obtain an estimate $\tilde{g}(k)$ of $g(k)$ for $k=2^d$ using $M_d=\alpha + \gamma (d_f+1-d)$ repetitions of both experiments.
        \label{step:singlephase_exePFE}
        \item Compute $\tilde{\theta}^{(d)}=\mathrm{Arg}[\tilde{g}(2^d)]\in[0,2\pi)$ \label{step:singlephase_deftheta}
        \item If $d=0$, set $\tilde{\phi}^{(0)}=\tilde{\theta}^{(0)}$.
        \item \label{step:update-QPE} Else, set $\tilde{\phi}^{(d)}$ to be the unique value in the interval  $[\tilde{\phi}^{(d-1)}-\frac{\pi}{2^d},\tilde{\phi}^{(d-1)}+\frac{\pi}{2^d})$ (with periodic boundaries) such that 
        \begin{equation}
        2^d\tilde{\phi}^{(d)} = \tilde{\theta}^{(d)}\mod 2\pi.
        \end{equation}
    \end{enumerate}
    \item Return $\tilde{\phi}=\tilde{\phi}^{(d_f)}$ as an estimate for $\phi$.
\end{enumerate}
It was proven in \cite{Kimmel15Robust} that for some choices of $\alpha$ and $\gamma$ the root-mean-square error $\delta$ on the final estimate $\phi^{(d_f)}$ is at most $c T^{-1}$ for a constant $c$ and total cost $T=2\sum_{d=0}^{d_f} 2^d M_d$, thus reaching the Heisenberg limit.
\end{alg}

One might consider the effect of experimental noise on these limits.
The algorithm given in \cite{Kimmel15Robust} was proven to be robust against noise that affected the estimation of any $g(k)$ by no more than $\frac{1}{\sqrt{8}}$.
However, realistic noise tends to scale with the circuit depth, eventually breaking this bound.
In the presence of a uniform depolarizing channel, it is possible to extend the above calculation of Fisher information to optimize the recovery of a single phase $\phi$, however in the limit that $\delta\rightarrow 0$ only the sampling-noise limit can be obtained:
\begin{lem}\label{lem:depol_channel}
The (root-mean-square) error $\delta$ of an estimator $\tilde{\phi}$ of the eigenvalue phase $\phi$ employing the circuits in Fig.~\ref{fig:QPE} on an eigenstate of $U$ in the presence of a pure depolarizing channel with a fixed failure probability $p=e^{-1/\tau}$ per iteration of $U$ is bounded as
\begin{equation}
    \delta\geq \frac{\tau^{-\frac{1}{2}}T^{-\frac{1}{2}}}{2e}.
\end{equation}
\end{lem}
\begin{proof}
A pure depolarizing channel sends the off-diagonal element of the reduced density matrix in Eq.~\eqref{eq:control_rdm} to $p^kg(k)$.
This adjustment can be propagated directly through to the Fisher information (Eq.~\eqref{eq:Fisher_info_single_eigenstate}), which becomes
\begin{equation}
    I(\phi)=\sum_ke^{-2k/\tau}k^2(M_k^r+M_k^i),
\end{equation}
while the total quantum cost (Eq.~\ref{eq:Total_quantum_cost_single_eigenstate}) remains the same.
Let us consider the relative contribution to $I$ versus the contribution to $T$ of a single choice of $k$; if we write $I(\phi)=\sum_kI_k$ and $T=\sum_kT_k$, we have $I_k/T_k=ke^{-2k/\tau}$.
Differentiating w.r.t. $k$ and setting equal to zero yields
\begin{align}
    \frac{d}{dk}\frac{I_k}{T_k}=e^{-2k/\tau}-\frac{2k}{\tau}e^{-2k/\tau}=0\\
    &\rightarrow k=\frac{\tau}{2}.
\end{align}
Optimizing $I(\phi)$ with respect to $T$ thus requires setting $k=\frac{\tau}{2}$ and increasing $M_k=M_k^r+M_k^i$, which yields
\begin{equation}
    I(\phi)=\frac{\tau^2}{4e}M_k, T=\frac{\tau}{2}M_k.
\end{equation}
Substituting this into the Cram\'{e}r-Rao bound $\delta\geq I(\phi)^{-\frac{1}{2}}$ yields the desired result.
\end{proof}
{\em Remark}: Note that as we fix $T\sim\tau$ in the above, we in effect have Heisenberg-limited scaling in $\tau$.
However, if we treat $\tau$ as a constant this is only the sampling noise limit as defined above.
This result holds only for the simplest-possible noise case; more complicated noise is difficult to analyse, but numerical results show it may prevent estimation beyond some minimum value using standard techniques~\cite{Obrien19Quantum}.
We assume herein that all circuits are noiseless (and will not use Lemma~\ref{lem:depol_channel} in the rest of this work).

\section{Defining the task of multiple-eigenvalue phase estimation}\label{sec:QPE_sparse_def}

In this section we define the goal of estimating multiple eigenvalue phases of some unitary $U$. When the input state $\ket{\Psi}$ is supported on multiple eigenstates, choosing a single $k$ does not suffice, simply since knowing Eq.~\eqref{eq:phase_function} at a single point $k$ does not give a unique solution $\{A_j,\phi_j\}$ \cite{Obrien19Quantum}.
A simple way to circumvent this problem is thus to estimate `densely': estimating $g(k)$ for all integers $0\leq k\leq K$ would allow us to fit up to $O(K)$ $(\phi_j,A_j)$ pairs. However, this does not saturate the Heisenberg limit as shown in Theorem \ref{thm:heisenberg_limit}, hence we need to come up with a different method.

Separate from this, the full eigenspectrum of an arbitrary $N$-qubit unitary $U$ has up to $2^N$ unique values, making it impossible to describe in polynomial time in $N$. In addition, the spectral content of the input state $\ket{\Psi}$ could be very dense, with many eigenvalues clustered together instead of separated by gaps, and the overlap for these eigenvalues, $A_j$, could be sharply concentrated or uniformly spread. To deal with general input states, Ref.~\cite{Somma19Quantum} thus formulated the quantum eigenvalue estimation problem (QEEP): instead of estimating individual phases, the focus is on estimating the spectral function $A(\phi)$ in Eq.~\eqref{eq:spectral_function} with some resolution.
We recall the precise definition of this problem in App.~\ref{app:QEEP}, Def.~\ref{def:QEEP}.
In our case, we focus on the case where the initial state only has a non-zero overlap with a small number $n_\phi$ of eigenvectors of $U$, and we want to estimate eigenphases corresponding to each of them.
Here is our precise definition of the problem to be solved:

\begin{dfn}[Multiple eigenvalue estimation problem]\label{def:MEEP} Fix an error bound $\delta>0$, an overlap bound $A>0$. For a unitary $U$ and state $|\Psi\rangle$, we assume that $A_j>A$ for exactly $n_{\phi}$ phases $\phi_j$ and $A_j=0$ for all other phases so that $n_{\phi}\leq A^{-1}$. 
Let $g(k)=\sum_jA_je^{ik\phi_j}$ be the phase function in Def.~\ref{def:phase-function} and assume access to the PFE quantum subroutine in Def.~\ref{def:phase_function_estimation_problem} for any $k\in \mathbb{R}^+$, generating data ${\bf x}$. The task is to output a set $\{\tilde{\phi}_l\}$ of $n_{\phi}$ or fewer estimates of the phases $\{\phi_j\}$ such that, if we take the closest estimate $\tilde{\phi}_j^{(\mathrm{closest})}({\bf x})$ of each phase $\phi_j$ given the data ${\bf x}$,
\begin{equation}
    \tilde{\phi}_j^{(\mathrm{closest})}({\bf x})=\argmin_{\tilde{\phi}_l}(|\tilde{\phi}_l({\bf x})-\phi_j|_T),
\end{equation}
the accuracy error
{\small
\begin{equation}
    \delta_j=\sqrt{\sum_{\mathbf{x}}\mathbb{P}\left(\mathbf{x}|\{\phi_l,A_l\}\right)\left|\tilde{\phi}_j^{(\mathrm{closest})}(\mathbf{x})-\phi_j\right|_T^2},
    \label{eq:MSE}
\end{equation}}
is bounded by $\delta_j \leq \delta$ for all $j=1, \ldots, n_{\phi}$.
\end{dfn}

{\em Remark}: Def.~\ref{def:MEEP} allows us the freedom to assign a single estimate to multiple phases when calculating the final mean-squared-error. 

\subsection{Methods of dense signal phase estimation}\label{sec:dense_phase_methods}

Achieving the Heisenberg limit for multiple eigenvalues requires solving the problem considered in Def.~\ref{def:MEEP} with a total quantum cost $T = O(\delta^{-1})$.
We intend to accomplish that goal with a multi-order estimation scheme.
At each order $d$, we will use a data processing method to estimate multiple eigenphases $\theta^{(d)}_j$ of $U^{k_d}$ to within some error $\epsilon$, from data generated by PFE (analogous to step \ref{step:singlephase_deftheta} of Alg.~\ref{alg:single_phase_Heisenberg_limit}).
(We will stitch the estimates of the phases $\theta^{(d)}_j$ together to give Heisenberg-limited estimates of the corresponding $\phi_j$ in a manner similar to step \ref{step:update-QPE} of Alg.~\ref{alg:single_phase_Heisenberg_limit}.)
We can offload the estimation of $\theta^{(d)}_j$ to a subroutine; we will show later that we can afford a subroutine with superlinear scaling in $\epsilon$ as the final error $\delta$ in our multi-order scheme can be made arbitrarily small even at fixed $\epsilon$.
In this section we discuss the two subroutines that we will consider in this work: the matrix pencil method (Alg.~\ref{alg:mps}) first studied for QPE in Ref.~\cite{Obrien19Quantum}, and the `time series analysis' proposed in Ref.~\cite{Somma19Quantum} to solve the QEEP problem mentioned above.

We detail our implementation of the matrix pencil method in Alg.~\ref{alg:mps}; this is a well-known algorithm in signal processing, that is known to achieve the dense-sampling limit for a single eigenvalue~\cite{Rao45Information, Hua90Matrix}.
However, bounding the error of the matrix pencil method in estimating many phases typically requires a minimal gap $\Delta$ between these phases, $\Delta=\min_{i\neq j} |\phi_i-\phi_j|_T$, and that we query the PFE to obtain estimates of $g(k)$ at $k>\frac{1}{\Delta}$.
This is a proven necessary condition to estimate multiple $\phi_j$ to error $\epsilon\leq c\Delta$~\cite{Moitra14Super} for some constant $c$.
In our case, we need to allow for the case where we are estimating two phases $\theta_0^{(d)}$, $\theta_1^{(d)}$ to an error $\epsilon\geq |\theta_0^{(d)}-\theta_1^{(d)}|$.
In principle this is not forbidden by the result of~\cite{Moitra14Super} (and numerical simulation confirms that this indeed works), but we do not know of a formal statement about the scaling of the matrix pencil method in this situation.
Instead, we opt for a different method for the purposes of forming a rigorous proof of the Heisenberg limit, and test the matrix pencil method in numerics only.

Ref.~\cite{Somma19Quantum} proved rigorously the QEEP can be solved from the densely sampled signal $g(k)$ generated with the {\rm PFE} in Def.~\ref{def:phase_function_estimation_problem} using a `time-series analysis' algorithm.
We review the results of Ref.~\cite{Somma19Quantum} in detail in App.~\ref{app:QEEP}.
However, the solution to the QEEP is an estimation $\tilde{A}(\phi)$ of a discretization of the spectral function $A(\phi)$ (Eq.~\eqref{eq:spectral_function}) rather than a set of estimates $\{\tilde{\theta}_j^{(d)}\}$.
To use the time-series analysis algorithm as a subroutine in our multi-order phase estimation algorithm then requires converting from one form to the other.
This is achieved by Alg.~\ref{alg:phase_extraction}, the Conservative QEEP Eigenvalue Extraction algorithm.
This algorithm is designed so that its output fulfills the following guarantees whenever the time-series analysis algorithm succeeds (defined as $\|\tilde{A}(\phi)-A(\phi)\|_1\leq\epsilon$)

\begin{lem}\label{lem:phase_extraction_promises}
Fix a confidence bound $0<p<1$, an overlap bound $A$, a number of phases $n_{\theta} < \frac{1}{A}$, and an error bound $0 < \epsilon < \frac{A}{3}$.
Let $g(k)=\sum_jA_je^{ik\theta_j}$ be the phase function for a unitary $V$, with $A_j >A$ for exactly $n_\theta$ phases $\theta_j$, and $A_j = 0$ for all other phases. 
Let $\{\tilde{\theta}_l\}$ be a set of estimates of $\{\theta_j\}$ generated by Alg.~\ref{alg:phase_extraction} with error $\epsilon$ and confidence bound $p$.
With probability at least $p$, the following statements are true:
\begin{enumerate}
    \item \label{stat1} For each phase $\theta_j$ with $A_j>0$, there exists at least one estimate $\tilde{\theta}_l$ such that $$|\theta_j-\tilde{\theta}_l|_T\leq 2\epsilon.$$
    \item \label{stat2} For each estimate $\tilde{\theta}_l$ there exists at least one phase $\theta_j$ with $A_j>0$ such that $$|\theta_j-\tilde{\theta}_l|_T\leq 2\epsilon.$$
    \item \label{stat3} The number of estimates $|\{\tilde{\theta}_l\}|\leq n_{\theta}$.
\end{enumerate}
\end{lem}
See App~\ref{app:QEEP} for a proof.
The total quantum cost of Alg.~\ref{alg:phase_extraction} is inherited directly from the cost of the time series analysis algorithm (as it involves no additional quantum circuitry), which is $O(\epsilon^{-6}|\log(1-p)|)$.

\section{Multiple eigenvalues: multi-order estimation and the phase matching problem}\label{sec:phase_matching_problem}

To achieve Heisenberg-limited scaling for multiple phases, we combine the dense signal algorithms which can resolve multiple phases in the previous section with the single-phase Heisenberg limited algorithm, Algorithm~\ref{alg:single_phase_Heisenberg_limit}, which achieves the correct scaling.

A natural way to achieve such combination is to estimate phases $\{\theta^{(d)}_j\}$ of $V=U^{2^d}$ for multiple orders $d=0,1\ldots, d_f$ via a dense signal method (e.g. Alg.~\ref{alg:phase_extraction}), and then combine them in the same manner as in Algorithm ~\ref{alg:single_phase_Heisenberg_limit}. 

If we would manage to get an estimate of $\theta_j^{(d)}=\phi_j 2^d$ at each order $d$ with error $\epsilon$ {\em and} be able to combine these estimates in an unequivocal manner, then reaching the Heisenberg limit for multiple phases may be feasible.
Note that the error in the final $d_f^{\rm th}$ estimate in this case would be $\delta\sim \epsilon / 2^{d_f}$ with $\epsilon$ in Algorithm \ref{alg:phase_extraction}. One could thus achieve arbitrarily small $\delta$ for fixed $\epsilon$ by making $d_f$ arbitrarily large. This allows us to use (possibly non-optimal) routines such as the QEEP algorithm since the scaling with $\epsilon$ does not propagate to a scaling in $\delta$ for a sufficiently small~$\epsilon$.

However, the combination of phase information at different orders in the case of multiple phases may not be feasible when we use $V=U^{2^d}$ for increasing $d$ as in Algorithm \ref{alg:single_phase_Heisenberg_limit}. The reason is that if we have multiple phases, the previous order estimates provide sets of `ballpark' intervals and it may not clear or unambiguous which interval to choose in order to convert a new estimate $\tilde{\theta}^{(d)}_j$ to an updated $\tilde{\phi}^{(d)}_j$ (as in step \ref{step:update-QPE} of Algorithm \ref{alg:single_phase_Heisenberg_limit}). We would like the choice of the next order to be such that this `matching with a previous estimate' can be done unambiguously.

For this, we will estimate the multiple eigenphases of $V=U^{k_d}$ for $k_d=\prod_{d'=1}^d\kappa_{d'}$ with $\kappa_d\geq 2$ a, possibly non-integer, multiplier.
For this algorithm to have a means of associating each $d$th order estimate $\tilde{\theta}^{(d)}_j$ with a single previous-order estimates $\tilde{\phi}^{(d-1)}_j$, we use an adaptive strategy for choosing the next multiplier $\kappa_d$ in Alg.~\ref{alg:adaptive}. That is, the algorithm will determine a $\kappa_d$ based on the estimates $\tilde{\phi}_j^{(d-1)}$ from the previous round. Although this scheme requires some classical processing of the experimental data before the experiment is finished, it is not an adaptive scheme in the same sense as iterative QPE~\cite{Kitaev95Quantum}, as we do not require feedback within the coherent lifetime of a single experiment.

The generalization from using $U^{2^d}$ to $U^{k_d}$ for $k_d\in\RR^+$ presents one small additional complication. In order to prove bounds on the estimation at each order we will require invoking Lemma~\ref{lem:normscaling}.
However, this requires that our phases $\phi_j$ satisfy Eq.~\eqref{eq:normscaling_condition} (unless $k_d\in\mathbb{Z}^+$~). If a phase $\phi_j$ does not satisfy Eq.~\eqref{eq:normscaling_condition}, one can construct a situation where two corresponding estimates $\tilde{\phi}_j^{(d)}$ are found on either side of the branch cut at $2\pi$, and where we cannot guarantee that our algorithm would choose the `correct' one (without knowledge of the hidden $\phi_j$). To solve this issue, we note that one may shift the phases of $U$ by a constant $\chi$ by performing phase estimation on $Ue^{-i\chi}$ instead of $U$.
This need not even be done on the quantum device, as one simply multiplies estimates of $g(k)$ by $e^{-ik\chi}$.
As we assume the existence of only $n_{\phi}$ phases, we can always find some $Ue^{i\chi}$ with phases in some window $[\phi_{\min}, \phi_{\max})$ with $\phi_{\min}\geq \frac{\pi}{k}$, $\phi_{\max}\leq \frac{\pi(2\lfloor k\rfloor -1)}{k}$ when $k \geq 3n_{\phi}$.
This will allow us to invoke Lem.~\ref{lem:normscaling} to match estimates of eigenphases of $Ue^{i\chi}$ and estimates of eigenphases of $(Ue^{i\chi})^k$ as we require. We also note that the above issue can be circumvented when $U=e^{iHt}$ by a suitable choice of $t$.

\subsection{Heisenberg-limited algorithm for multiple phases}
\label{sec:toHeis}

We now describe our Heisenberg-limited phase estimation algorithm.
This algorithm targets a final error $\delta=O(\delta_c)$, where $\delta_c$ is a fixed input to the algorithm itself (We will calculate the constant of proportionality in the proof of Theorem \ref{thm:adaptive_gets_Heisenberg_limit}).
The Heisenberg limit will be achieved by making this $\delta_c$ smaller while keeping the error $\epsilon$ of the phase extraction subroutine, Alg. \ref{alg:phase_extraction}, constant. 

\begin{alg}\label{alg:adaptive}[Adaptive multi-order phase estimation algorithm]
We assume access to the conservative QEEP eigenvalue extraction algorithm, Alg.~\ref{alg:phase_extraction} for a unitary $V=U^k$ (for arbitrary $k \in\RR^+$), and an initial state $|\Psi\rangle$. Fix a final error $\delta_c$, an overlap bound $A$, a number of phases $n_{\phi}\leq A^{-1}$, and error parameters $\epsilon_0$ and $\epsilon$ bounded as 
\begin{equation}
\label{eq:eps-n0}
    \epsilon_0 \leq \epsilon_{\mathrm{crit},0} \equiv \frac{2 \pi}{300 n_{\phi}^4}
\end{equation}
and
\begin{equation}
\label{eq:eps-n}
    \epsilon \leq \epsilon_{{\rm crit}}\equiv \frac{2\pi}{300 n_{\phi}^2}.
\end{equation}
Let the confidence parameter $p_d$ be
\begin{equation}
        p_d=1-e^{-\alpha}\left(\frac{k_d\delta_c}{\pi}\right)^{\gamma},
        \label{eq:defpd}
    \end{equation}
given some real numbers $\alpha> 0$, $\gamma>2$ and $k_d$ to be chosen below. The algorithm proceeds as follows:
\begin{enumerate}
    \item \label{step:order0} Let $d=0$ and $k_{d=0}=1$. Use Alg.~\ref{alg:phase_extraction} to find a set of first estimates $\{\tilde{\phi}^{(0)}_j\}$ of eigenvalues of $U$ with error parameter $\epsilon_0$ in Eq.~\eqref{eq:eps-n0}, overlap bound $A$, and confidence $p_{d=0}$ in Eq. \eqref{eq:defpd}. If this set is empty or has more than $n_{\phi}$ elements, return $\{0\}$ (this is a failure mode).
    \item \label{step:shift} Find the point $\zeta \in [0,2\pi)$ defined by
    \begin{equation}
        \zeta=\argmax_{\zeta'\in[0,2\pi]}\;\min_j\left|\tilde{\phi}^{(0)}_j-\zeta'\right|_T,
    \end{equation}
    i.e. $\zeta$ is the midway point in the largest gap between the phase estimates $\tilde{\phi}^{(0)}_j$. Let 
    \begin{equation}
        d_{\zeta} = \min_j\left|\tilde{\phi}^{(0)}_j-\zeta\right|_T,
    \end{equation}
    i.e. $d_{\zeta}$ is half the size of the largest gap. Shift the unitary $U\rightarrow Ue^{-i(\zeta+\frac{1}{2}d_{\zeta}-8\epsilon_0)}$, $\tilde{\phi}_j^{(0)}\rightarrow\tilde{\phi}^{(0)}_j-\zeta-d_{\zeta}/2+8\epsilon_0\mod 2\pi$.
    \item \label{step:k1} Choose $\kappa_1=k_1$ with $k_1 \in [3n_\phi, 3n_{\phi}+1]$ such that for all $\tilde{\phi}^{(0)}_j \neq \tilde{\phi}^{(0)}_l$, either
        \begin{equation}
        |\tilde{\phi}_j^{(0)}k_1 - \tilde{\phi}_l^{(0)}k_1|_T > 4\epsilon_0(1+ k_1).~\label{eq:first_round_aliasing}
    \end{equation}
     or
    \begin{equation}
        |\tilde{\phi}_j^{(0)}-\tilde{\phi}_l^{(0)}|_T < \frac{\pi}{k_1}.~\label{eq:first_round_aliasing_lower_bound}
    \end{equation}
    \item \label{step:loop} While $k_d<\frac{2\epsilon}{\delta_c}$ :
    \begin{enumerate}
    \item Set $d\rightarrow d+1$.
    \item \label{step:exeQEEP} Use Alg.~\ref{alg:phase_extraction} to find a set of estimates $\{\tilde{\theta}^{(d)}_l\}$ of eigenvalues of $V=U^{k_d}$ with error parameter $\epsilon$ in Eq.~\eqref{eq:eps-n}, overlap bound $A$, and confidence $p_d$ in Eq.~\eqref{eq:defpd}.
    \item  \label{step:fail1} If there exists some $\tilde{\phi}_j^{(d-1)}$ such that
    {\small
    \begin{equation}
    \min_l|k_d\tilde{\phi}_j^{(d-1)}-\tilde{\theta}_l^{(d)}|_T>2\epsilon(1+\kappa_d),\label{eq:far1}
    \end{equation}
    }
    or there exists some $\tilde{\theta}_l^{(d)}$ such that
    {\small
    \begin{equation}
        \min_j|k_d\tilde{\phi}_j^{(d-1)}-\tilde{\theta}_l^{(d)}|_T>2\epsilon(1+\kappa_d) \label{eq:far2},
    \end{equation}}
    or the number of estimates $|\{\tilde{\theta}_l^{(d)}\}|>n_{\phi}$, return $\{\tilde{\phi}_j^{(d-1)} + \zeta+d_{\zeta} / 2 -8\epsilon_0 \mod 2\pi\}$. This is a failure mode. 
    \item  \label{step:update} If not, for each $\tilde{\theta}_l^{(d)}$, find the estimate $\tilde{\phi}_j^{(d-1)}$ and an integer $n\in[0,k_d)$ which minimizes
    \begin{equation}
        |\tilde{\phi}_j^{(d-1)} - (\tilde{\theta}_l^{(d)}+2\pi n)/k_d|_T,\label{eq:cost_function_nj}
    \end{equation}
    and set $\{\tilde{\phi}_l^{(d)}\}=\{(\tilde{\theta}^{(d)}_l+2\pi n)/k_d\}$.
    \item \label{step:fail2} If any $\tilde{\phi}_j^{(d)}\in[0,\frac{\pi}{k_d})\cup (\frac{\pi(2\lfloor k_d\rfloor - 1)}{k_d}, 2\pi]$, return $\{\tilde{\phi}_j^{(d-1)} + \zeta+d_{\zeta} / 2 - 8\epsilon_0 \mod 2\pi\}$. This is a failure mode.
      \item \label{step:krest} Choose the multiplier $\kappa_{d+1}\in [2,3]$ such that for all $\tilde{\phi}^{(d)}_j \neq \tilde{\phi}^{(d)}_l$, either
    \begin{align}
        |\tilde{\phi}_j^{(d)}k_d\kappa_{d+1} - \tilde{\phi}_l^{(d)}k_{d}\kappa_{d+1}|_T\nonumber\\
         > 4\epsilon(1+\kappa_{d+1}).~\label{eq:aliasing_bound}
    \end{align}
     or
     {\small
    \begin{equation}
        |\tilde{\phi}_j^{(d)}-\tilde{\phi}_l^{(d)}|_T < \frac{\pi-2\epsilon(1+\kappa_{d+1})}{k_d\kappa_{d+1}},~\label{eq:aliasing_lower_bound}
    \end{equation}}
    and set $k_{d+1}=k_{d}\kappa_{d+1}$.
    \end{enumerate}
    \item Return $\{\tilde{\phi}_j^{(d)} + \zeta+d_{\zeta} / 2 - 8\epsilon_0 \mod 2\pi\}$.
\end{enumerate}
\end{alg}
In principle the first few orders $d$ could be skipped given accurate prior knowledge of our phases $\phi_j$.
However, as the largest circuits are executed at the latter $d$ values, this will only change the constant factor of the algorithm, rather than the asymptotic scaling with $\delta$.

In the rest of this section, we prove that Alg.~\ref{alg:adaptive} can achieve the Heisenberg limit.
The first use of the QEEP subroutine requires a potentially smaller error parameter ($\epsilon_0$, bounded by Eq.~\eqref{eq:eps-n0}) than subsequent uses (where $\epsilon$ needs to be only bounded by Eq.~\eqref{eq:eps-n}).
(Invoking Alg.~\ref{alg:phase_extraction} requires that the error parameters $\epsilon$ and $\epsilon_0$ are at most $A/3 \leq 1/(3 n_{\phi})$, which is fulfilled by both bounds.)
This relates to a technical issue: we require $k_1\geq 3 n_{\phi}$ in order for Lemma \ref{lem:shifting_unitary} in Appendix \ref{app:shift} and thus Lemma \ref{lem:normscaling} to apply.
For later rounds $k_d \geq 3 n_{\phi}$ automatically, and the multiplier $\kappa_d$ is no longer constrained, which indirectly allows us to relax the region of valid choices for $\epsilon$. 
The first step in proving the performance of Algorithm \ref{alg:adaptive} is to show that the multipliers can be chosen in the first (step \ref{step:k1} in Alg. \ref{alg:adaptive}), and subsequent rounds (step \ref{step:krest} in Alg. \ref{alg:adaptive}), which obey the desired conditions.
This is accomplished by the following Lemma which is proved in Appendix \ref{app:phase_matching_proof}. 
\begin{lem}\label{lem:phase_matching_solution}
Let $\{\tilde{\phi}_j^{(0)}\}\in [0,2\pi)$ be a set of at most $n_{\phi}$ phases. Assuming Eq.~\eqref{eq:eps-n0}, for a randomly chosen $k_1  \in [3n_\phi, 3 n_{\phi}+1]$ with probability at least $1/2$, either Eq.~\eqref{eq:first_round_aliasing} or Eq.~\eqref{eq:first_round_aliasing_lower_bound} holds for all $\tilde{\phi}_j^{(0)}\neq \tilde{\phi}_l^{(0)}$. Fix a $k_d$. Let $\{\tilde{\phi}^{(d)}_j\}\in[0,2\pi)$ be a set of at most $n_{\phi}$ phases. Assuming Eq.~\eqref{eq:eps-n}, for a randomly chosen $\kappa_{d+1}\in [2,3]$ with probability at least $3/4$, either Eq.~\eqref{eq:aliasing_bound} or Eq.~\eqref{eq:aliasing_lower_bound}  holds for all $\tilde{\phi}_l^{(d)} \neq \tilde{\phi}_j^{(d)}$.
\end{lem}

{\em Remarks}: The probability with which a multiplier can be found which obeys the desired property is rather arbitrary in this Lemma and can be increased by choosing a smaller $\epsilon$. Note that it is easy to verify whether for a randomly chosen multiplier the desired conditions hold or not. The validity of this Lemma importantly does not depend on whether the phase estimates are actually accurate, it only depends on the number of phases $n_{\phi}$. In practice, we do not generate a random multiplier $\kappa_{d+1}$ through this Lemma, but simply exhaustively search for a valid $\kappa_{d+1}$ starting at the maximal value, see Section \ref{sec:numerics}.

The reason to adaptively choose the multiplier $\kappa_{d+1}$ for $d=0,1,\ldots$ is that two (estimated) phases in principle need to lead to separate estimates at the next order: this is expressed in Eq.~\eqref{eq:aliasing_bound}. An exception to this occurs when the (estimated) phases are still close enough, as in Eq.~\eqref{eq:aliasing_lower_bound}, so that their next-order refined estimates could merge at the next order, see Fig.~\ref{fig:aliasing}. Phase estimates can thus split and merge over the multiple orders. They split when sufficient accuracy is available at the next order to distinguish them, they can stay or are allowed to merge when such accuracy is not yet needed at the given order. 

In what follows below we will assume, just for simplicity of the proof, that the error parameter $\epsilon$ is bounded by $\epsilon_{{\rm crit},0}$ in Eq.~\eqref{eq:eps-n0} for all rounds, and $\epsilon$ is the same for all rounds, including the first one.\\

\begin{figure*}[ht!]
    \centering
    \includegraphics[width=\textwidth]{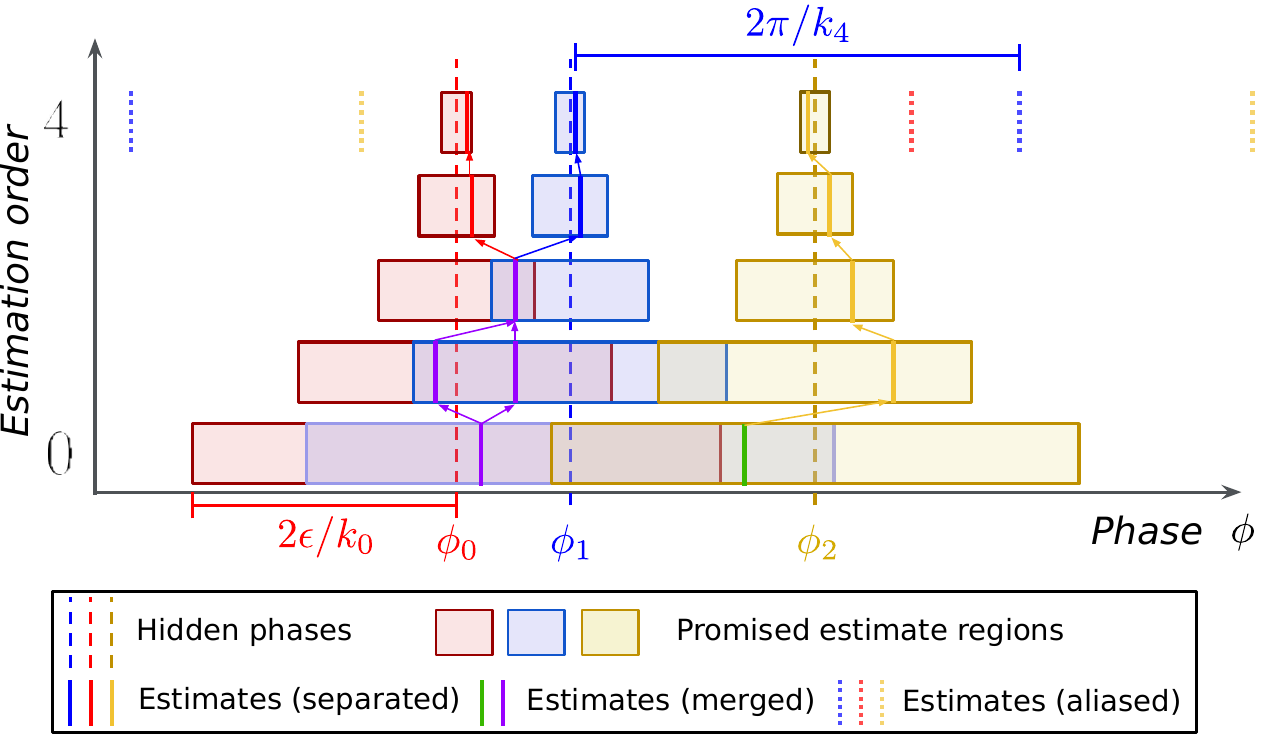}
    \caption{Schematic of the execution of Alg.~\ref{alg:adaptive} to estimate three hidden phases, $\phi_0$, $\phi_1$ and $\phi_2$ (dashed lines). The algorithm progresses from bottom to top as the estimation order $d$ increases. At each order $d$, the phase extraction subroutine (Alg.~\ref{alg:phase_extraction}) promises to return an estimate $\tilde{\theta}^{(d)}_{\rightarrow j}$ of each $\theta^{(d)}_j=k_d\phi_j\mod 2\pi$ that corresponds to an estimate $\tilde{\phi}^{(d)}_{\rightarrow j}$ (solid lines) lying within the promised estimate region about $\phi_j$ (coloured boxes). By matching phases at subsequent orders (arrows), the algorithm is able to converge to an ever-more accurate estimate of each $\phi_j$. The phase extraction subroutine only promises that each region will contain at least one phase (and that the total number of estimates at each order is bounded by $n_{\phi}$) - when two regions overlap, the subroutine may merge the phases to give a single estimate (green and purple lines). Estimates at subsequent orders may continue to separate and even re-merge until the regions separate, at which point the algorithm promises with high confidence a precise estimate for each hidden phase. The estimate $\tilde{\phi}_{\rightarrow j}^{(d)}$ at each order is only known $\mod (2\pi)/k_d$, leading to a set of potential aliases (dotted lines at $d=4$) for each phase. We do not know a priori which alias is correct, and must rely on the fact that the true estimate $\tilde{\phi}_{\rightarrow j}^{(d)}$ needs to be close to a previous estimate $\tilde{\phi}_{\rightarrow j}^{(d-1)}$. By carefully choosing each $k_d$, we can guarantee that no alias will satisfy this condition (so long as Alg.~\ref{alg:phase_extraction} succeeds), and our phase matching will be unambiguous.}
    \label{fig:aliasing}
\end{figure*}

\subsection{Bounding the error with and without failures}

In this section we state and prove the two key intermediate lemmas, Lemma \ref{lem:phase_matching_checks} and Lemma \ref{lem:phase_matching_with_failures} on our way towards proving that Alg.~\ref{alg:adaptive} reaches the Heisenberg-limit. Together, these lemmas allow us to bound the error in Alg.~\ref{alg:adaptive}, --assuming that the phase extraction subroutine succeeds for the first $d$ rounds--, to within $O(\epsilon/k_d)$.

These Lemmas deal with the issue of `aliasing' or the correct matching of new estimates with older estimates which is solved by the specific choice of $\kappa_{d+1}$ in step \ref{step:krest} of Alg.~\ref{alg:adaptive}, see also Fig.~\ref{fig:aliasing}. 
It is important to note that there is no 1-1 relation between these estimates and the actual phases since the number of estimates is at most the number of phases.

Let $d_f$ be the last order executed in Alg.~\ref{alg:adaptive}, i.e. the last order for which we go through step \ref{step:exeQEEP}, construct the estimates $\{\tilde\theta_l^{(d_f)}\}$ and pass the tests at step \ref{step:fail1} and \ref{step:fail2} and output $\{\tilde\phi_l^{(d_f)}\}$. When none of the failure modes is encountered, $d_f$ is set by the first $k_d$ such that $k_{d} \geq \frac{\epsilon}{\delta_c}$ (since the next $k_{d+1}\geq 2\epsilon/\delta_c$ as $\kappa_{d+1}\geq 2$). Since $\kappa_d \geq 2$, we observe that 
\begin{equation}
    d_f \leq \log_2\left( \frac{2\epsilon}{\delta_c}\right).
    \label{eq:upper-df}
\end{equation}
In Corollary \ref{cor:fail} we argue that when the QEEP subroutines, Alg. \ref{alg:phase_extraction}, succeed up to order $d_f$, we indeed never exit via these failure modes. 

\begin{lem}\label{lem:phase_matching_checks}
If each invocation of the QEEP subroutine, Alg.~\ref{alg:phase_extraction}, succeeds in Alg.~\ref{alg:adaptive} up to order $d_f$, then in this last round $d_f$ in step \ref{step:update} it holds that 
\begin{itemize}
    \item (Property 1a) For every phase $\phi_j$ there exists an estimate $\tilde{\phi}_l^{(d_f)}$ such that
    $$|\phi_j-\tilde{\phi}_l^{(d_f)}|_T \leq 2\epsilon/k_{d_f}.$$ 
    \item (Property 1b) For every estimate $\tilde{\phi}_l^{(d_f)}$ there exists a phase $\phi_j$ such that 
    $$|\phi_j-\tilde{\phi}_l^{(d_f)}|_T \leq 2\epsilon/k_{d_f}.$$
\end{itemize}
\end{lem}

\begin{proof} We prove this Lemma by induction. Consider the first round $d=0$ ($k_{d=0}=1$), i.e. step \ref{step:order0} of Alg. \ref{alg:adaptive}. If the QEEP subroutine, Alg.~\ref{alg:phase_extraction}, succeeds (with probability $p_0$) then Lemma \ref{lem:phase_extraction_promises} holds, namely for each $\phi_j$ there exists an estimate $\tilde{\phi}_l^{(0)}$ such that 
\begin{equation}
\label{eq:d0-a}
    |\phi_j-\tilde{\phi}_l^{(0)}|_T \leq 2\epsilon.
\end{equation}
and for each estimate $\tilde{\phi}_l^{(0)}$ there exists at least one $\phi_j$ such that 
\begin{equation}
\label{eq:d0-b}
    |\phi_j-\tilde{\phi}_l^{(0)}|_T \leq 2\epsilon.
\end{equation}
Hence the statement to be proven holds at $d=0$. Now consider step \ref{step:shift} of Alg. \ref{alg:adaptive} and invoke Lemma \ref{lem:shifting_unitary} for which the assumptions are fulfilled by Eqs.~\eqref{eq:d0-a},\eqref{eq:d0-b}. This implies that with the choice of $k_1 \geq 3 n_{\phi}$ in step \ref{step:k1} in Alg.~\ref{alg:adaptive} the shifted phases and their 0th-order estimates obey the technical condition in Lemma \ref{lem:normscaling} and we can use Eq.~\eqref{eq:normscaling}. In the next steps we work with these shifted phases but for simplicity we don't use any new notation and refer to them as $\phi_j$ and estimates $\tilde{\phi}_j^{(d)}$ etc.

Now assume the statement to be proven holds at order $d$ i.e. 
let $\{\tilde{\phi}_l^{(d)}\}$ be a set of at most $ n_{\phi}$ estimates of the phases $\{\phi_j\}$ with
\begin{itemize}
\item (Assumption 1a) For every phase $\phi_j$ there exists an estimate $\tilde{\phi}_l^{(d)}$ such that
$$|\phi_j-\tilde{\phi}_l^{(d)}|_T \leq 2\epsilon/k_{d}.$$
    \item (Assumption 1b) For every estimate $\tilde{\phi}_l^{(d)}$ there exists a phase $\phi_j$ such that
    $$|\phi_j-\tilde{\phi}_l^{(d)}|_T \leq 2\epsilon/k_{d}.$$
\end{itemize}
Note that these assumptions certainly imply that one can apply Lemma \ref{lem:normscaling} to the estimates $\tilde{\phi}_l^{(d)}$. That is, given that the real phases $\phi_j$ are $2\epsilon/k_d$ close to these estimates and that the (shifted) $\phi_j$ obey Eq.~\eqref{eq:realphases}, it implies that Eq.~\eqref{eq:normscaling} can be used with $k\geq 3 n_{\phi}$ (which is the case for all rounds $d\geq 1$).

We consider the QEEP subroutine, Alg. \ref{alg:phase_extraction}, with a given choice of $\kappa_d$ obeying the conditions in step \ref{step:k1} (for $d=1$) and step \ref{step:krest} (for higher $d$), executed in step \ref{step:exeQEEP} with confidence $p_d$. In the math below we refer to conditions on $\kappa_{d> 1}$, namely Eq.~\eqref{eq:aliasing_bound} and Eq.~\eqref{eq:aliasing_lower_bound}, but the conditions on $\kappa_1$ in Eq.~\eqref{eq:first_round_aliasing} and Eq.~\eqref{eq:first_round_aliasing_lower_bound} are of identical form (so we don't make a separate argument for the $d=0 \rightarrow d=1$ induction step).

Let thus $\{\tilde{\theta}^{(d+1)}_l\}$ be a set of estimates of the eigenphases $\{\theta^{(d+1)}_j\}$ of $U^{k_{d+1}}$ corresponding to the set $\{\phi_j\}$, that is, 
\begin{equation}
    \theta^{(d+1)}_j = k_{d+1}\phi_j\mod 2\pi.\label{eq:phitheta}
\end{equation}
and $k_{d+1}=k_d \kappa_{d+1}$. By assuming that Alg.~ \ref{alg:phase_extraction} succeeds we can invoke Lemma \ref{lem:phase_extraction_promises}, namely
\begin{itemize}
    \item (Assumption 2a) For every phase $\theta_j^{(d+1)}$ there exists an estimate $\tilde{\theta}_l^{(d+1)}$ such that $$|\theta_j^{(d+1)}-\tilde{\theta}_l^{(d+1)}|_T \leq 2\epsilon.$$
    \item (Assumption 2b) For every estimate $\tilde{\theta}_l^{(d+1)}$ there exists a phase $\theta_j^{(d+1)}$ such that $$|\theta_j^{(d+1)}-\tilde{\theta}_l^{(d+1)}|_T \leq 2\epsilon.$$
    \end{itemize}
To prove the induction step, we thus need to show that the set $\tilde{\phi}_j^{(d+1)}$ generated by step \ref{step:update} of Alg.~\ref{alg:adaptive} satisfies the following two properties
\begin{itemize}
    \item (Property 1a) For every phase $\phi_j$ there exists an estimate $\tilde{\phi}_l^{(d+1)}$ such that
    $$|\phi_j-\tilde{\phi}_l^{(d+1)}|_T \leq 2\epsilon/k_{d+1}.$$
    \item (Property 1b) For every estimate $\tilde{\phi}_l^{(d+1)}$ there exists a phase $\phi_j$ such that
    $$|\phi_j-\tilde{\phi}_l^{(d+1)}|_T \leq 2\epsilon/k_{d+1}.$$
\end{itemize}

First consider Assumption 2a. Assumption 2a implies that for every phase $\phi_j$ there exists a $\tilde{\theta}_{l}^{(d+1)}$ such that
\begin{equation}
    |k_{d+1} \phi_j-\tilde{\theta}_{l}^{(d+1)}|_T \leq 2\epsilon.
    \label{eq:ass}
\end{equation}
In this proof we will use the label $l=\rightarrow j$ for this $\tilde{\theta}_l^{(d+1)}$ associated with $\phi_j$. Thus, also using Eq.~\eqref{eq:normscaling}, Eq.~\eqref{eq:ass} is equivalent to
\begin{align}
    \min_{n\in \{0, \ldots, \lfloor k_{d+1}\rfloor-1\}} |\phi_j-(\tilde{\theta}_{\rightarrow j}^{(d+1)}+&2\pi n)/k_{d+1}|_T\nonumber\\ &\leq \frac{2\epsilon}{k_{d+1}},
    \label{eq:bound-wanted}
\end{align}
with $n^{\rm ideal}_{j,\rightarrow j}$ the integer which achieves the minimum, i.e.
{\small
\begin{align}
 n^{\rm ideal}_{j,\rightarrow j}=&
    \argmin_{n\in \{0, \ldots, \lfloor k_{d+1}\rfloor-1\}}\Big\{\nonumber\\
 & |\phi_j-(\tilde{\theta}_{\rightarrow j}^{(d+1)}+2\pi n)/k_{d+1}|_T\Big\}.
\end{align}}
Similarly, by Assumption 1a, there is some $\tilde{\phi}_{l=\rightarrow j}^{(d)}$ which is $2\epsilon/k_d$-close to $\phi_j$, again using a label which shows this association.

Consider the optimization in Eq.~\eqref{eq:cost_function_nj} at step \ref{step:update} in Alg.~\ref{alg:adaptive}. We define
\begin{align}
\xi_l&=\min_j \xi_{j,l},\\   \xi_{j,l} &\equiv \left|\tilde{\phi}^{(d)}_j-(\tilde{\theta}_l^{(d+1)}+2\pi n_{j,l})/k_{d+1}\right|_T
\end{align}
with
\begin{align}
    n_{j,l} =& \argmin_{n\in \{0, \ldots, \lfloor k_{d+1}\rfloor-1\}}\Big\{\nonumber\\&\left|\tilde{\phi}^{(d)}_j-(\tilde{\theta}_l^{(d+1)}+2\pi n)/k_{d+1}\right|_T\Big\},\\
    n_l =&\argmin_{n_{j,l}} \xi_{j,l}. \label{eq:njd_def}
\end{align}
The goal is thus to prove that for each $\phi_j$, using the corresponding $\tilde{\theta}_{\rightarrow j}^{(d+1)}$, we have $n_{\rightarrow j}=n^{\rm ideal}_{j,\rightarrow j}$ which directly implies Property 1a.

We can bound using Eq.~\eqref{eq:triangle_inequality} and then Eq.~\eqref{eq:normscaling}, Assumptions 1a and 2a and the optimality of $n_{\rightarrow j, \rightarrow j}$,
\begin{align}
    &\xi_{\rightarrow j,\rightarrow j}\nonumber\\&\hspace{0.2cm}=\left|\tilde{\phi}^{(d)}_{\rightarrow j}-(\tilde{\theta}_{\rightarrow j}^{(d+1)}+2\pi n_{\rightarrow j, \rightarrow j})/k_{d+1}\right|_T \notag \\
    &\hspace{0.2cm}\leq\left|\tilde{\phi}^{(d)}_{\rightarrow j}-\phi_j\right|_T\nonumber \notag \\
    &\hspace{0.2cm}\qquad +\left|\phi_{j}-(\tilde{\theta}_{\rightarrow j}^{(d+1)}+2\pi n^{\rm ideal}_{j, \rightarrow j})/k_{d+1}\right|_T \notag \\
    &\hspace{0.2cm} \leq \frac{2\epsilon}{k_d}+\frac{1}{k_{d+1}}\left|k_{d+1}\phi_j-\tilde{\theta}^{(d+1)}_{\rightarrow j}\right|_T \notag \\
    &\hspace{0.2cm}= \frac{2\epsilon}{k_d} + \frac{1}{k_{d+1}}\left|\theta_j^{(d+1)}-\tilde{\theta}_{\rightarrow j}^{(d+1)}\right|_T \notag \\
    &\hspace{0.2cm} \leq \frac{2\epsilon(1 + \kappa_{d+1})}{k_{d+1}}.\label{eq:boundinglemma_phijthetajnjbound1}
\end{align}

Now if Eq.~\eqref{eq:aliasing_bound} holds for some other $m \neq \rightarrow j$, we claim on the other hand that
\begin{equation}
\xi_{m,\rightarrow j} > \frac{2\epsilon(1 + \kappa_{d+1})}{k_{d+1}},\label{eq:boundinglemma_step_1_mid2}
\end{equation}
hence matching $\theta_{\rightarrow j}^{(d+1)}$ with such $\tilde{\phi}_m^{(d)}$, with $m \neq \rightarrow j$ is non-optimal and will not be chosen in the Algorithm. To indeed see that Eq.~\eqref{eq:aliasing_bound} implies Eq.~\eqref{eq:boundinglemma_step_1_mid2}, we can calculate
{\small
\begin{align}
    &\frac{4\epsilon (1+ \kappa_{d+1})}{k_{d+1}}<\frac{1}{k_{d+1}}\left|k_{d+1}\tilde{\phi}_{\rightarrow j}^{(d)}-k_{d+1}\tilde{\phi}_m^{(d)}\right|_T \notag \\
    &\hspace{2cm}\leq\frac{1}{k_{d+1}}\left|k_{d+1}\tilde{\phi}_{\rightarrow j}^{(d)}-\tilde{\theta}_{\rightarrow j}^{(d+1)}\right|_T\nonumber\\
    &\hspace{2cm}+\frac{1}{k_{d+1}}\left|\tilde{\theta}_{\rightarrow j}^{(d+1)}-k_{d+1}\tilde{\phi}_m^{(d)}\right|_T \notag \\
    &\leq \frac{2\epsilon (1+ \kappa_{d+1})}{k_{d+1}}+\min_{n\in\{0,\ldots,\lfloor k_{d+1}\rfloor -1\}}\Big\{\nonumber\\&\hspace{2cm}\Big|\tilde{\phi}_m^{(d)}-(\tilde{\theta}_{\rightarrow j}^{(d+1)}+2\pi n)/k_{d+1}\Big|_T\Big\} \\
&\Rightarrow    \frac{2\epsilon (1+ \kappa_{d+1})}{k_{d+1}} <\nonumber\\
    &\min_{n\in\{0,\ldots,\lfloor k_{d+1}\rfloor -1\}}\left|\tilde{\phi}_m^{(d)}-(\tilde{\theta}_{\rightarrow j}^{(d+1)}+2\pi n)/k_{d+1}\right|_T \notag \\
    &\hspace{4cm}=\xi_{m, \rightarrow j}.
\end{align}}
Here we have used that Eq.~\eqref{eq:normscaling} holds for the estimate $\tilde{\phi}_m^{(d)}$. 

Alternatively, for those $m \neq \rightarrow j$ for which Eq.~\eqref{eq:aliasing_lower_bound} holds, we claim that
\begin{equation}
    n_{m,\rightarrow j}= n_{\rightarrow j,\rightarrow j},\label{eq:boundinglemma_step_1_mid}
\end{equation}
hence for those $\tilde{\phi}^{(d)}_m \neq \tilde{\phi}^{(d)}_{\rightarrow j}$ the algorithm produces a single new estimate equal to $(\tilde{\theta}_{\rightarrow j}^{(d+1)}+2\pi n_{\rightarrow j,\rightarrow j})/k_{d+1}$.

To see that Eq.~\eqref{eq:aliasing_lower_bound} implies Eq.~\eqref{eq:boundinglemma_step_1_mid} indeed, note that it is sufficient to prove that
{\small
\begin{equation}
    \left|\tilde{\phi}_m^{(d)}-(\tilde{\theta}_{\rightarrow j}^{(d+1)}+2\pi n_{\rightarrow j,\rightarrow j})/k_{d+1}\right|_T<\frac{\pi}{k_{d+1}},
    \label{eq:intermed}
\end{equation}}
as one can then show that for $n'\neq n_{\rightarrow j,\rightarrow j}\in \{0,\ldots,\lfloor k_{d+1}\rfloor -1\}$ that
{\small
\begin{align}
     &\frac{2\pi}{k_{d+1}}\leq \left|\frac{2\pi}{k_{d+1}}(n_{\rightarrow j,\rightarrow j}-n')\right|_T \notag \\
     &\hspace{0.5cm}= |(\tilde{\theta}_{\rightarrow j}^{(d+1)}+2\pi n_{\rightarrow j,\rightarrow j})/k_{d+1}\nonumber\\
    &\hspace{0.5cm}\qquad\qquad-(\tilde{\theta}_{\rightarrow j}^{(d+1)}+2\pi n')/k_{d+1}|_T \notag \\
     &\hspace{0.5cm}\leq \left|(\tilde{\theta}_{\rightarrow j}^{(d+1)}+2\pi n_{\rightarrow j,\rightarrow j})/k_{d+1}-\tilde{\phi}_m^{(d)}\right|_T\nonumber\\
    &\hspace{0.5cm}\qquad+\left|\tilde{\phi}_m^{(d)}-(\tilde{\theta}_{
     \rightarrow j}^{(d+1)}+2\pi n')/k_{d+1}\right|_T \notag \\
     &\hspace{0.5cm}< \frac{\pi}{k_{d+1}}+\left|\tilde{\phi}_m^{(d)}-(\tilde{\theta}_{\rightarrow j}^{(d+1)}+2\pi n')/k_{d+1}\right|_T \\
  &\Rightarrow   \frac{\pi}{k_{d+1}}< \left|\tilde{\phi}_m^{(d)}-(\tilde{\theta}_{\rightarrow j}^{(d+1)}+2\pi n')/k_{d+1}\right|_T,
     \label{eq:same-n}
\end{align}}
so $n_{\rightarrow j,\rightarrow j}$ is optimal.
Using Eq.~\eqref{eq:triangle_inequality}, Eq.~\eqref{eq:aliasing_lower_bound} and Eq.~\eqref{eq:boundinglemma_phijthetajnjbound}, we can prove Eq. \eqref{eq:intermed} since
\begin{align}
&\left|\tilde{\phi}_m^{(d)}-(\tilde{\theta}_{\rightarrow j}^{(d+1)}+2\pi n_{\rightarrow j,\rightarrow j})/k_{d+1}\right|_T \notag \\
&\hspace{0.2cm}\leq \left|\tilde{\phi}_m^{(d)}-\tilde{\phi}_{\rightarrow j}^{(d)}\right|\notag \\
&\hspace{0.4cm}+\left|\tilde{\phi}_{\rightarrow j}^{(d)}-(\tilde{\theta}_{\rightarrow j}^{(d+1)}+2\pi n_{\rightarrow j,\rightarrow j})/k_{d+1}\right|_T \notag \\
&\hspace{0.2cm}< \frac{\pi-2\epsilon(1+\kappa_{d+1})}{k_{d+1}}+\frac{2\epsilon(1+\kappa_{d+1})}{k_{d+1}} \notag \\
&\hspace{4cm}= \frac{\pi}{k_{d+1}}.
\label{eq:final}
\end{align}

We have thus shown that for each $\phi_j$, there is a $\tilde{\theta}_{\rightarrow j}^{(d+1)}$, such that step \ref{step:update} will output $(\tilde{\theta}_{\rightarrow j}^{(d+1)}+2\pi n_{\rightarrow j,\rightarrow j})/k_{d+1}$ with $n_{\rightarrow j,\rightarrow j}$ defined in Eq.~\eqref{eq:njd_def}, related to the previous order estimate $\phi_{\rightarrow j}^{(d)}$ which was already close to $\phi_j^{(d)}$. The last step is to show that $n_{\rightarrow j, \rightarrow j}=n_{j,\rightarrow j}^{\rm ideal}$ using Property 1a. It holds that  
\begin{align}
    &|\phi_{j}-(\tilde{\theta}_{\rightarrow j}^{(d+1)}+2\pi n_{\rightarrow j, \rightarrow j})/k_{d+1}|_T \notag \\
    &\hspace{0.2cm}\leq  |\phi_{j}-\tilde{\phi}^{(d)}_{\rightarrow j}|_T \notag \\
    &\hspace{0.5cm}+ |\tilde{\phi}^{(d)}_{\rightarrow j}- (\tilde{\theta}_{\rightarrow j}^{(d+1)}+2\pi n_{\rightarrow j, \rightarrow j})/k_{d+1}|_T \notag \\
&\hspace{0.2cm}\leq \frac{2\epsilon}{k_d} + \frac{2\epsilon (1+\kappa_{d+1})}{k_{d+1}}
< \frac{\pi}{k_{d+1}}, 
\label{eq:dev}
\end{align}
where we used that $4 \epsilon (\kappa_{d+1}+1) < \pi$. Indeed for $d >1$, $\epsilon < \frac{\pi}{16}$ ($\kappa_{d+1} \leq 3$) and for $d=0$, $\epsilon < \frac{\pi}{4(3 n_{\phi}+2)}$ ($\kappa_1 \leq 3 n_{\phi}+1$) given the upper bounds on $\epsilon$ in Eqs.~\eqref{eq:eps-n} and \eqref{eq:eps-n0}. This implies through the same argument as in Eq.~\eqref{eq:same-n} that $n_{j,\rightarrow j}^{\rm ideal}$ achieving the minimum in Eq.~\eqref{eq:bound-wanted} equals $n_{\rightarrow j, \rightarrow j}$ and hence we obtain Property 1a. 

Now let's prove Property 1b. Given a $\tilde{\theta}_l^{(d+1)}$, let $\phi_{j=\rightarrow l}$ be the real phase for which, by Assumption 2b, it holds that
\begin{multline}
    |k_{d+1}\phi_{\rightarrow l}-\tilde{\theta}_l^{(d+1)}|_T\\=k_{d+1} |\phi_{\rightarrow l}-(\tilde{\theta}_{l}^{(d+1)}+2\pi n^{\rm ideal}_{\rightarrow l, l})/k_{d+1}|_T\\
    \leq 2\epsilon.
    \label{eq:propbstart}
\end{multline}
To prove Property 1b, we need to show that $n_l=n^{\rm ideal}_{\rightarrow l, l}$ with $n_l$ defined in Eq.~\eqref{eq:njd_def}. Let also $\tilde{\phi}_{\rightarrow j}^{(d)}$ be the previous order $2\epsilon/k_d$-close estimate to $\phi_{j=\rightarrow l}$ by Assumption 1a. The idea is that $\tilde{\phi}_{\rightarrow j}^{(d)}=\tilde{\phi}_{\rightarrow (\rightarrow l)}^{(d)}$ is matched to $\tilde{\theta}_l^{(d+1)}$ in the optimization step of the algorithm, so that this leads to a better estimate for the phase $\phi_{j=\rightarrow l}$.

Given a $\tilde{\theta}_l^{(d+1)}$ we can deduce, as before, that
\begin{align}
    \xi_{\rightarrow j,l}&=\left|\tilde{\phi}^{(d)}_{\rightarrow j}-(\tilde{\theta}_{l}^{(d+1)}+2\pi n_{\rightarrow j, l})/k_{d+1}\right|_T \notag \\
    &\leq\left|\tilde{\phi}^{(d)}_{\rightarrow j}-\phi_{\rightarrow l}\right|_T\nonumber\\
    &\hspace{0.1cm}+\left|\phi_{\rightarrow l}-(\tilde{\theta}_{l}^{(d+1)}+2\pi n^{\rm ideal}_{\rightarrow l, l})/k_{d+1}\right|_T \notag \\
    & \leq  \frac{2\epsilon}{k_d}+\frac{1}{k_{d+1}}\left|k_{d+1}\phi_{\rightarrow l}-\tilde{\theta}^{(d+1)}_{l}\right|_T \notag \\
    &= \frac{2\epsilon}{k_d} + \frac{1}{k_{d+1}}\left|\theta_{\rightarrow l}^{(d+1)}-\tilde{\theta}_{l}^{(d+1)}\right|_T \notag \\
    &\leq \frac{2\epsilon(1 + \kappa_{d+1})}{k_{d+1}}.\label{eq:boundinglemma_phijthetajnjbound}
\end{align}
Using previous arguments, all other $\xi_{m,l}$ are either larger or give the same integer $n_{\rightarrow j,l}$ and thus $n_l=n_{\rightarrow j,l}$. In addition, we can bound, using this equality and Assumption 1a
 \begin{multline}
    |\phi_{\rightarrow l}-(\tilde{\theta}_{l}^{(d+1)}+2\pi n_{\rightarrow j, l})/k_{d+1}|_T
    \\
    \leq  |\phi_{\rightarrow l}-\tilde{\phi}^{(d)}_{\rightarrow j}|_T + |\tilde{\phi}^{(d)}_{\rightarrow j}- (\tilde{\theta}_l^{(d+1)}+2\pi n_{\rightarrow j, l})/k_{d+1}|_T 
    \\
\leq \frac{2\epsilon}{k_d} + \frac{2\epsilon (1+\kappa_{d+1})}{k_{d+1}}\\
< \frac{\pi}{k_{d+1}}, 
\end{multline}
implying that $n_l=n_{\rightarrow j,l}=n^{\rm ideal}_{\rightarrow l,l}$ as desired.
 \end{proof}

Algorithm \ref{alg:adaptive} has a few failure modes, namely steps \ref{step:order0}, \ref{step:fail1} and \ref{step:fail2} where we exit and return an estimate of lower order. Arguments in Lemma \ref{lem:phase_matching_checks} show that these failure modes are only encountered when the QEEP subroutine, Alg. \ref{alg:phase_extraction}, fails at some order.
`Failure' here is not complete failure; regardless of whether the algorithm fails, it will return a set of estimates $\tilde{\phi}_j$, and the error in these estimates will contribute to Eq.~\ref{eq:MSE}.
To achieve the Heisenberg limit we must make sure that both the probability of failure is small, and that the estimates $\tilde{\phi}_j$ from a failed instance of the algorithm still lie close to the true values $\phi_j$ to minimize their contribution to Eq.~\eqref{eq:MSE}.
The probability $p_d$ with which the QEEP subroutine succeeds at the $d$th order is bounded by the parameters $\alpha, \gamma$, given as an input to Alg.~\ref{alg:adaptive} (Eq.~\ref{eq:defpd}).
The probability that this achieves for up to and including the $d$th order is bounded by $\prod_{d'=0}^d p_{d'}$.
It is crucial for the success of our algorithm that we do not encounter any exit modes other than those mentioned above, which we can now prove given the machinery developed in Lem.~\ref{lem:phase_matching_checks}.

\begin{cor}
If each invocation of the QEEP subroutine, Alg.~\ref{alg:phase_extraction}, succeeds in Alg.~\ref{alg:adaptive}, we never exit at step \ref{step:order0}, \ref{step:fail1} or \ref{step:fail2}.
\label{cor:fail}
\end{cor}

\begin{proof}
Consider step \ref{step:order0} of Alg. ~\ref{alg:adaptive} applying Alg.~\ref{alg:phase_extraction} which obeys Lemma \ref{lem:phase_extraction_promises}, showing that success of Alg. \ref{alg:phase_extraction} implies that the number of phases is at most $n_{\phi}$. By assumption there is at least one phase with $A_j > 0$, and hence success means that Alg.~\ref{alg:phase_extraction} cannot return the empty set due to statement 1. of Lemma \ref{lem:phase_extraction_promises}. Hence if Alg.~\ref{alg:phase_extraction} succeeds we do not exit at step \ref{step:order0} of Alg.~\ref{alg:adaptive}. Now consider step \ref{step:fail1} of Alg.~\ref{alg:adaptive}: again success of Alg.~\ref{alg:phase_extraction} implies that the number of estimates does not exceed $n_{\phi}$. Consider Eq.~\eqref{eq:far1} and Eq.~\eqref{eq:far2}; we wish to show that these will not hold if Alg.~\ref{alg:phase_extraction} succeeds up to order $d-1$.
If Alg.~\ref{alg:phase_extraction} succeeds up to order $d-1$, the phase estimates $\tilde{\phi}_j^{(d-1)}$ obey Eq.~\eqref{eq:normscaling} for $k_d$, hence the condition in Eq.~\eqref{eq:far2} equals, for each $\tilde{\theta}_l^{(d)}$
\begin{align}
    &\min_j \min_n |\tilde{\phi}_j^{(d-1)}-(\tilde{\theta}_l^{(d)}+2\pi n)/k_d|_T\nonumber\\&\hspace{4cm}>\frac{2\epsilon(1+\kappa_d)}{k_d},
\end{align}
and we argued previously, via induction, that this does not happen when Alg.~\ref{alg:phase_extraction} succeeds up to order $d$, as $\min_j \xi_{j,l}$ is upper-bounded as in Eq.~\eqref{eq:boundinglemma_phijthetajnjbound} for all $d' \leq d$. 
Similarly, Eq.~\eqref{eq:far1} implies the existence of a $\tilde{\phi}_j^{(d-1)}$ with
\begin{align}
    &\min_l \min_n |\tilde{\phi}_j^{(d-1)}-(\tilde{\theta}_l^{(d)}+2\pi n)/k_d|_T\nonumber\\&\hspace{4cm}>\frac{2\epsilon(1+\kappa_d)}{k_d}.
\end{align}
which can not happen due to the success of Alg.~\ref{alg:phase_extraction} which implies the bound in Eq.~\eqref{eq:boundinglemma_phijthetajnjbound1}.
Consider lastly step \ref{step:fail2} which exits if the current $d$th order estimates do not lie in the region for which Eq.~\eqref{eq:normscaling_condition} holds with given $k_d$. We have argued in Lemma \ref{lem:phase_matching_checks} that, assuming success of the subroutines implementing Alg.~\ref{alg:phase_extraction}, that Eq.~\eqref{eq:normscaling_condition} holds for the phase estimates at all orders.
\end{proof}

Now let us consider failures of the QEEP subroutine, Alg. \ref{alg:phase_extraction}, which do not lead to exiting. Let's imagine that the first failure occurs at some order $d_0$. Now we want to make sure that continuing with higher orders after such failure still leads to an error of order $\sim \epsilon/k_{d_0-1}$, even though the failure (or any subsequent failure) is not detected.

To show this, we check that if Alg.~\ref{alg:adaptive} exits at some later round, namely during $d=d_{f}+1$ and outputs estimates $\tilde{\phi}_j^{(d_f)}$ that these will be sufficiently close to the estimates right before failure, that is, the set of phases $\tilde{\phi}_j^{(d_0-1)}$. 

Then, by Lem.~\ref{lem:phase_matching_checks}, these estimates will also be sufficiently close to the true phases $\phi_j$.

\begin{lem}
\label{lem:phase_matching_with_failures}
Let Alg.~\ref{alg:adaptive} exit at order $d=d_f+1$ and let the QEEP subroutine, Alg.~\ref{alg:phase_extraction}, of step \ref{step:exeQEEP} first fail at $d=d_0 \leq d_f+1$.
For each $\phi_j$, there will be an estimate $\tilde{\phi}_l^{(d_f)}$, produced at step \ref{step:update} in Alg.~\ref{alg:adaptive} which satisfies
\begin{equation}
    \left|\phi_j-\tilde{\phi}_l^{(d_f)}\right|_T \leq  \frac{14\epsilon}{k_{d_0-1}}.
    \label{eq:claim1}
\end{equation}
Vice-versa, for each estimate $\tilde{\phi}_l^{(d_f)}$ there exists a phase $\phi_j$ such that 
\begin{equation}
    \left|\phi_j-\tilde{\phi}_l^{(d_f)}\right|_T \leq  \frac{14\epsilon}{k_{d_0-1}}.
    \label{eq:claim2}
\end{equation}
\end{lem}

\begin{proof} 
Since each QEEP subroutine, Alg. \ref{alg:phase_extraction}, in Alg.~\ref{alg:adaptive} succeeds up to order $d_0-1$, Lemma \ref{lem:phase_matching_checks} guarantees that
\begin{itemize}
    \item (Property 1a) For every phase $\phi_j$ there exists an estimate $\tilde{\phi}_l^{(d_0-1)}$ such that $$|\phi_j-\tilde{\phi}_l^{(d_0-1)}|_T \leq \frac{2\epsilon}{k_{d_0-1}}.$$
    \item (Property 1b) For every estimate $\tilde{\phi}_l^{(d_0-1)}$ there exists a phase $\phi_j$ such that $$|\phi_j-\tilde{\phi}_l^{(d_0-1)}|_T \leq \frac{2\epsilon}{k_{d_0-1}}.$$
\end{itemize}
Then, since the algorithm does not exit at step \ref{step:fail1}  through Eqs.~\eqref{eq:far1}  or step \ref{step:fail2} for any order $d=d_0, \ldots d_f$, for each estimate $\tilde{\phi}_l^{(d-1)}$ we can associate some $\tilde{\theta}_{m_l}^{(d)}$ that satisfies 
\begin{multline}
    \frac{1}{k_{d}}\left|k_{d}\tilde{\phi}_l^{(d-1)}-\tilde{\theta}_{m_l}^{(d)}\right|_T =\\
    \min_{n \in \{0,\ldots, \lfloor k_{d+1}\rfloor - 1\}} \left|\tilde{\phi}_l^{(d-1)}-(\tilde{\theta}_{m_l}^{(d)}+2\pi n)/k_{d}\right|_T \\
    \leq \frac{2\epsilon(1 + \kappa_{d})}{k_{d}},
    \label{eq:close}
\end{multline}
where the second equality follows from being allowed to apply Eq.~\eqref{eq:normscaling} (which is validated by passing the test at step \ref{step:fail2}). 
This implies that in step \ref{step:update} of Alg.~\ref{alg:adaptive} at round $d$, for a given $\tilde{\theta}_{m_l}^{(d)}$, the optimization of $\xi_{n,m_l}$ over $n$ will pick the integer $n_{l,m_l}$, i.e. the integer associated with matching $\tilde{\theta}_{m_l}^{(d)}$ with $\tilde{\phi}_l^{(d-1)}$. 
Next, similar as in the proof of Lemma \ref{lem:phase_matching_checks}, we can consider the possibility of matching to other estimates $\tilde{\phi}_k^{(d-1)}\neq \tilde{\phi}_l^{(d-1)}$. Since $\kappa_{d}$ is chosen in step \ref{step:krest} of Alg.~\ref{alg:adaptive}, we claim that either Eq.~\eqref{eq:aliasing_bound} holds, in which case
\begin{multline}
    \label{eq:step2_case1}
    \min_{n
    \in\{0,\ldots,\lfloor k_{d}\rfloor - 1\}}\left|\tilde{\phi}_k^{(d-1)}-\frac{2\pi n+\tilde{\theta}_{m_l}^{(d)}}{k_{d}}\right|_T\\
    > \frac{2\epsilon}{k_{d}}(1 + \kappa_{d}),
\end{multline}
hence this $\xi_{k,m_l}$ is not optimal, or that Eq.~\eqref{eq:aliasing_lower_bound} holds, in which case
\begin{multline}
    \label{eq:step2_case2}
   n_{l,m_l}=\\ \argmin_{n\in \{0,\ldots,\lfloor k_{d}\rfloor - 1\}}\left|\tilde{\phi}_k^{(d-1)}-\frac{2\pi n+\tilde{\theta}_{m_l}^{(d)}}{k_{d}}\right|_T\\
   = n_{k,m_l}.
\end{multline}
The proofs of these claims are exactly the same as in the proof of Lemma \ref{lem:phase_matching_checks}, i.e. using Eqs.~\eqref{eq:intermed}, \eqref{eq:same-n},\eqref{eq:final}.

Now let us prove Eq.~\eqref{eq:claim1}. Given a phase $\phi_j$, we can use Property (1a) to find an associated estimate $\tilde{\phi}_{\rightarrow j}^{(d_0-1)}$ within $2\epsilon/k_{d_0-1}$. Then for this estimate let $\theta^{(d_0)}_{m_{\rightarrow j}}$ be the matched estimate in the next round for which Eq.~\eqref{eq:close} holds, so that the round produces a new estimate $\tilde{\phi}^{(d_0)}_{\rightarrow j}=\frac{\tilde{\theta}^{(d_0)}_{m_{\rightarrow j}}+2 \pi n_{\rightarrow j, m_{\rightarrow j}}}{k_{d_0}}$ (which we label with $\rightarrow j$ again) for which 
\begin{equation}
   |\tilde{\phi}_{\rightarrow j}^{(d_0-1)} -\tilde{\phi}^{(d_0)}_{\rightarrow j}|_T \leq \frac{2\epsilon}{k_{d_0}}(1 + \kappa_{d_0}).
\end{equation}
Then again for $\tilde{\phi}^{(d_0)}_{\rightarrow j}$ there exists some matching $\theta^{(d_0+1)}_{m_{\rightarrow j}}$ etc. and this generates a series of estimates $\tilde{\phi}^{(d)}_{\rightarrow j}$ up to order $d_f$.
For a given $\phi_j$ we can then bound, using this series of estimates and $\kappa_d\geq 2$ for all $d$,
\begin{align}
    &\left|\tilde{\phi}_j^{(d_f)}-\phi_j\right|_T\leq\left|\tilde{\phi}_{\rightarrow j}^{(d_0-1)}-\phi_j\right|_T \nonumber\\
    &\hspace{1cm}\qquad+ \sum_{d=d_0}^{d_f}\left|\tilde{\phi}_{\rightarrow j}^{(d-1)}-\tilde{\phi}_{\rightarrow j}^{(d)}\right|_T \notag \\
    &\hspace{0.5cm} \leq \frac{2\epsilon}{k_{d_0-1}}+\sum_{d=d_0}^{d_f}\frac{2\epsilon(1+\kappa_{d})}{k_{d}} \notag \\
    &\hspace{0.5cm}=\frac{2\epsilon}{k_{d_0-1}}\left(1+\sum_{d=d_0}^{d_f} \frac{1+\kappa_d}{\kappa_{d_0} \kappa_{d_0+1} \ldots \kappa_{d}}\right) \notag \\
    &\hspace{0.5cm} \leq  \frac{2\epsilon}{k_{d_0-1}}\left(1+\sum_{n=0}^{\infty}\frac{3}{2^n}\right) \notag \\
    &\hspace{0.5cm}= \frac{14\epsilon}{k_{d_0-1}}.
    \label{eq:sumcontrib}
\end{align}
Now let's prove Eq.~\eqref{eq:claim2} and start with an estimate $\tilde{\phi}_l^{(d_f)}$ which was obtained from some $\tilde{\theta}_{l}^{(d_f)}$ matched with a previous estimate $\tilde{\phi}_l^{(d_f-1)}$ (just for convenience we again use the same label) such that $|\tilde{\phi}_l^{(d_f)}-\tilde{\phi}_l^{(d_f-1)}|_T \leq \frac{2\epsilon}{k_{d_f}}(1 + \kappa_{d_f})$, using that we do not exit through Eq.~\eqref{eq:far2}. Then again for this previous estimate $\tilde{\phi}_l^{(d_f-1)}$ we can repeat the argument and create a sequence of estimates up to $\tilde{\phi}^{(d_0-1)}_l$. For the last estimate, we invoke Property (1b), namely that there is a nearby $\phi_j$. Then we can upperbound for this $\phi_j$: $|\tilde{\phi}_l^{(d_f)}-\phi_j|_T\leq |\tilde{\phi}_l^{(d_0-1)}-\phi_j|_T+\sum_{d=d_0}^{d_f}|\tilde{\phi}_l^{(d)}-\tilde{\phi}_l^{(d-1)}|_T$ etc., exactly as in Eq.~\eqref{eq:sumcontrib}, leading to Eq.~\eqref{eq:claim2}.
\end{proof}

\subsection{Algorithm \ref{alg:adaptive} achieves the Heisenberg limit}

We have seen that the success of the QEEP subroutines in Alg.~\ref{alg:adaptive} leads to an error scaling as $\epsilon/k_{d_f} \sim \delta_c$. Now we must choose the success probability $p_d$ of these subroutines in Eq.~\eqref{eq:defpd}, depending on $\alpha, \gamma$ so that the total mean-squared-error is bounded by some $\delta^2=O(\delta_c^2)$ while the quantum cost $T=O(\delta^{-1})$. We note that the next theorem contains no logarithmic factors in $\delta^{-1}$, as in \cite{LT:heisenberg}, but achieves pure Heisenberg scaling. 

\begin{thm}\label{thm:adaptive_gets_Heisenberg_limit}
Algorithm \ref{alg:adaptive} solves the multiple eigenvalue estimation problem in Def.~\ref{def:MEEP} with accuracy error $\delta$ and total quantum cost $T=O(\delta^{-1})$, given $A, n_{\phi}$ and a fixed $\epsilon_0$ and $\epsilon$ obeying Eqs.~\eqref{eq:eps-n} and \eqref{eq:eps-n0}, and some choice for the constants $\alpha>0$ and $\gamma>2$.
\end{thm}

{\em Remarks}: Note that the dependence on the number of phases $n_{\phi}$ is not made explicit in the statement of this Theorem, but this dependence will be polynomial in $n_{\phi}$, not necessarily a very low-order polynomial. This dependence comes through the choice for $\epsilon_0$ (and $\epsilon$) via Eq.~\eqref{eq:eps-n0} (resp. Eq.~\eqref{eq:eps-n}) which sets the error and thus the running time of the QEEP Algorithm \ref{thm:QEEP_oracle}.

\begin{proof}
Our proof is motivated by the analysis in \cite{Kimmel15Robust} for a single phase $\phi$ leading to Theorem \ref{alg:single_phase_Heisenberg_limit}. The idea is to bound the mean-squared-error in the final estimation of $\phi$ by summing over error contributions at each order $d$ at which the phase extraction subroutine may fail (with probability $1-p_d$).

In our case the multipliers $\kappa_d$ (and $k_d$) at each order are not fixed (as in Theorem \ref{alg:single_phase_Heisenberg_limit}) but depend on phase estimates at previous orders and thus measurement data at previous orders. Our confidence parameter $p_d$ in Eq.~\eqref{eq:defpd}, which determines the number of repeats of experiments, and hence the cost, in Alg.~\ref{alg:phase_extraction}, depends on $k_d$ and is thus a random variable depending on previous measurement data. All measurement data are denoted by ${\bf x}$ and thus we have random variables $k_d(\{\kappa_{d'}({\bf x})\}_{d'=1}^d)$ and $p_d(\{\kappa_{d'}({\bf x})\}_{d'=1}^d)$.

Consider the mean-squared-error $\delta_j^2$ for the $j$th phase $\phi_j$ in Eq.~\eqref{eq:MSE} in Definition \ref{def:MEEP}.
We have three error contributions to consider given a choice for the random variable $k_d$.
\begin{enumerate}
    \item With probability $1-p_0$ the subroutine Alg.~\ref{alg:phase_extraction} in Alg.~\ref{alg:adaptive} fails at step \ref{step:order0} ($d=0$). In this case, as we always return {\em some} estimate, $\delta_j$ is bounded for all $j$ by $\pi$.
    \item With probability at most $(1-p_{d_0})\prod_{d=0}^{d_0-1}p_d\leq 1-p_{d_0}=e^{-\alpha}\left(\frac{k_{d_0}\delta_c}{\pi}\right)^{\gamma}$, the subroutine Alg.~\ref{alg:phase_extraction} in Alg.~\ref{alg:adaptive} fails for the first time at some order $1 \leq d_0 \leq d_f$, and the algorithm proceeds in any way afterwards (by possibly exiting or not). In this case, Lemma \ref{lem:phase_matching_with_failures} bounds $\delta_j$ for all $j$ by $\frac{14\epsilon}{k_{d_0-1}}$ or Lemma \ref{lem:phase_matching_checks} bounds $\delta_j$ for all $j$ by $\frac{2\epsilon}{k_{d_0-1}}\leq \frac{14\epsilon}{k_{d_0-1}}$.
    \item With probability less than $\prod_{d=0}^{d_f}p_d<1$ the subroutine Alg.~\ref{alg:phase_extraction} in Alg.~\ref{alg:adaptive} succeeds up to the final round $d_f$, and Lemma~\ref{lem:phase_matching_checks} implies that $\delta_j \leq 2\epsilon/k_{d_f} \leq 2 \delta_c$ for all $j$ as $k_{d_f} \geq \epsilon/\delta_c$.
\end{enumerate}

We can now bound the mean-squared-error as a sum over the above three contributions weighted by their relevant overlaps:
\begin{widetext}
{\small
\begin{align}
    \delta_j^2&\leq (1-p_0)\pi^2 + \sum_{{\bf x}} \mathbb{P}({\bf x})\left[\sum_{d_0=1}^{d_f-1}\mathbb{P}(\kappa_1, \ldots, \kappa_{d_0}|{\bf x}) (1-p_{d_0}(\{\kappa_d'\}_{d'=1}^{d_0}))\left[\frac{14\epsilon}{ k_{d_0-1}(\{\kappa_d'\}_{d'=1}^{d_0-1})}\right]^2\right]+4\delta_c^2 \notag \\
    &=\pi^2 e^{-\alpha}\left(\frac{\delta_c}{\pi}\right)^{\gamma}
     + \sum_{{\bf x}} \mathbb{P}({\bf x})\left[\sum_{d_0=1}^{d_f-1}\mathbb{P}(\kappa_1, \ldots, \kappa_{d_0}|{\bf x})e^{-\alpha}\left(\frac{k_{d_0}\delta_c}{\pi}\right)^{\gamma} \frac{196\epsilon^2}{k_{d_0-1}^2} \right]+4\delta_c^2 \notag \\
    &\leq\pi e^{-\alpha}\left(\frac{\delta_c}{\pi}\right)^{\gamma} + 0.15\times \delta_c^2 \sum_{{\bf x}} \mathbb{P}({\bf x}) \mathbb{P}(\kappa_1, \ldots, \kappa_{d_0}|{\bf x})\sum_{d_0=1}^{d_f-1}\left(\frac{k_{d_0}\delta_c}{\pi}\right)^{\gamma - 2} + 4\delta_c^2.
    \label{eq:xdep}
\end{align}}
\end{widetext}
Here we have removed the dependency of $k_{d_0}$ and $k_{d_0-1}$ on the previous multipliers for notational simplicity. For $d=1$ we have $e^{-\alpha} \kappa_{d_0}^2 196 (\epsilon/\pi)^2=e^{-\alpha} k_1^2 196 (\epsilon_0/\pi)^2 \leq \frac{196 \times 16 \times 4}{(300)^2}\leq 0.15$ due to Eq.~\eqref{eq:eps-n0}. For $d> 1$, $e^{-\alpha} \kappa_{d_0}^2 196 (\epsilon/\pi)^2 \leq \frac{9 \times 196 \times 4}{(300)^2}=0.08$, due to Eq.~\eqref{eq:eps-n}.

To evaluate the middle term, we write $k_{d_0}=k_{d_f}\frac{k_{d_0}}{k_{d_f}}$, and note that as $k_d=\prod_{d'=1}^d\kappa_{d'}$, we have $\frac{k_{d_0}}{k_{d_f}}\leq 2^{d_0-d_f}$ as the multiplier $\kappa_d \geq 2$. As $k_{d_f}\leq \frac{2\epsilon}{\delta_c} < \frac{\pi}{2\delta_c}$, we have
{\small
\begin{align}
    \sum_{d_0=1}^{d_f-1}\left(\frac{k_{d_0}\delta_c}{\pi}\right)^{\gamma-2} &\leq  \frac{1}{2^{\gamma-2}}\sum_{d_0=1}^{d_f-1}(2^{\gamma-2})^{(d_0-d_f)} \notag \\
    &\leq \frac{ 2^{4-\gamma}}{2^{\gamma}-4},
\end{align}}
where the last inequality holds since $\gamma > 2$. By letting the upper bound be independent of the $\kappa_d$s, we can remove the dependence on ${\bf x}$ in Eq.~\eqref{eq:xdep}, using that $\sum_{{\bf x}} \mathbb{P}({\bf x}) \mathbb{P}(\kappa_1, \ldots, \kappa_{d_0}|{\bf x})=1$. This yields a final bound on $\delta_j$ of
{\small
\begin{equation}
    \delta_j^2\leq \delta_c^2 \left[\pi^{1-\gamma} e^{-\alpha}\delta_c^{\gamma-2} +\frac{0.15 \times 2^{4-\gamma}}{2^{\gamma}-4}+4]\right].
\end{equation}}
As $\gamma>2$, this scales as $\delta_c^2$ as $\delta_c\rightarrow 0$.

Let us now calculate the cost of executing Alg.~\ref{alg:adaptive} in terms of the number of unitary applications. Again this depends on the choice of multiplier $\kappa_d$ at each step. Let us fix a sequences of $k_d$s, and let $d_f$ be the final round of estimation in this algorithm, i.e. the final round for which we invoked the quantum subroutine in Alg.~\ref{alg:phase_extraction}.
At each order $d$ we use $V^{\sf k}=U^{k_d {\sf k}}$, where ${\sf k} = 0,1,\ldots, K$, with $2M_d$ samples where $K$ is a function of $\epsilon$ as in Theorem \ref{thm:QEEP_oracle}.
The cost of each experiment is $k_d {\sf k}$. 

We can calculate
\begin{align}
    T =& \sum_{d=0}^{d_f}\sum_{{\sf k}=1}^{K}(2M_d k_d {\sf k})\notag  \\
    = &  \sum_{d=0}^{d_f}M_d k_d K(K+1).
    \label{eq:T-upper}
\end{align}
The QEEP algorithm in Theorem \ref{thm:QEEP_oracle} requires $M_d=\tilde{O}( |\ln(1-p_d)|\epsilon^{-4})$ with $p_d$ in Eq.~\eqref{eq:defpd} and $K=\tilde{O}(\epsilon^{-1})$. 

We may bound
\begin{align}
    T &\leq \tilde{O}(\epsilon^{-6})\sum_{d=0}^{d_f} k_d \left|-\alpha + \gamma\ln\left(\frac{k_d\delta_c}{\pi}\right)\right| \notag \\
    &=\tilde{O}(\epsilon^{-6})\sum_{d=0}^{d_f} k_d \left[\alpha - \gamma\ln\left(\frac{k_d\delta_c}{\pi}\right)\right].
\end{align}
We again bound $k_d=\frac{k_d}{k_{d_f}}k_{d_f}\leq2^{d-d_f}\frac{\pi}{2\delta_c}$, which yields
\begin{align}
        T &\leq \delta_c^{-1} \tilde{O}(\epsilon^{-6})\Bigg[\frac{\pi}{2}\nonumber\\&\times\sum_{d=0}^{d_f} 2^{d-d_f} \big[\alpha - \gamma(d-d_f-1)\ln(2)\big]\Bigg] \notag  \\
        &\leq\delta_c^{-1} \tilde{O}(\epsilon^{-6})(\alpha+2\gamma \ln(2)).
\end{align}

Combining our bounds then yields
\begin{align}
    \delta&\leq\delta_c\left[e^{-\alpha}\pi^{1-\gamma}+4+\frac{0.15 \times 2^{4-\gamma}}{2^{\gamma}-4}\right]^{\frac{1}{2}}\\
    &\leq T^{-1}\left[e^{-\alpha}\pi^{1-\gamma}+4+\frac{0.15 \times 2^{4-\gamma}}{2^{\gamma}-4}\right]^{\frac{1}{2}}\nonumber \notag \\
    &\qquad \times \tilde{O}(\epsilon^{-6})[\alpha+2\gamma\ln(2)] \notag \\
    &=O(T^{-1}),
\end{align}
which is the Heisenberg limit. Note that a dependence on the number of phase $n_{\phi}$ enters the scaling of $\delta$ via $\epsilon$ which needs to be bounded as in Eqs.~\eqref{eq:eps-n0} and \eqref{eq:eps-n}.
\end{proof}

\section{Numerical implementation}\label{sec:numerics}

Thm.~\ref{thm:adaptive_gets_Heisenberg_limit} requires using the QEEP algorithm (Theorem \ref{thm:QEEP_oracle} and Algorithm \ref{alg:phase_extraction}) in order to obtain provable bounds. Instead of analytic bounds, we now turn to a numerical demonstration, giving the opportunity to implement and test Algorithm ~\ref{alg:adaptive} with a few modifications.
We test the algorithm using two different sub-routines, one based on the matrix pencil method~\cite{Hua90Matrix}, and one based on the QEEP time-series analysis of Theorem \ref{thm:QEEP_oracle}, as described in Algorithm \ref{alg:phase_extraction}.
Code to implement all simulations can be found at~\url{https://github.com/alicjadut/qpe}.

To improve the practical performance of Alg.~\ref{alg:adaptive}, we make the following two small changes.
Firstly, instead of choosing $\kappa_{d+1}$ in step \ref{step:krest} in the ranges declared in Lem.~\ref{lem:phase_matching_solution}, we choose the largest $\kappa_{d+1}$ consistent with Eq.~\eqref{eq:aliasing_bound} and Eq.~\eqref{eq:aliasing_lower_bound} for all $\tilde{\phi}_j^{(d)},\tilde{\phi}_l^{(d)}$.
We note that the maximum such $\kappa_{d+1}$ is bounded above by $\frac{\pi}{2\epsilon}-1$, as the left-hand side of Eq.~\eqref{eq:aliasing_bound} is bounded above by $2\pi$ and the left-hand side of Eq.~\eqref{eq:aliasing_lower_bound} is bounded below by $0$.
(In practice, tighter bounds can be found by checking the boundaries of the regions $R_{jl}^{(n)}$ defined in Eq.~\eqref{def:rjln}, and we find the largest possible $\kappa_{d+1}$ by iterating backwards through these boundaries till a gap is found.)
Secondly, as the bounds for $\epsilon$ and $\epsilon_0$ in Lem.~\ref{lem:phase_matching_solution} are rather loose, and our performance scales rather badly in both, we choose the largest $\epsilon=\epsilon_0$ that allows all simulations to find a value of $\kappa_{d+1}>2$ at each order.

When using the matrix pencil processing subroutine, we follow the implementation described in Ref.~\cite{Obrien19Quantum}:
\begin{alg}\label{alg:mps}
The matrix pencil method takes as input estimates of the phase function $g({\sf k})=\sum_j A_j e^{i {\sf k} \theta_j}$ for a unitary $V$ at points ${\sf k}=0,1, \ldots K$ and an overlap bound $A$, and proceeds as follows:
\begin{enumerate}
\item Construct the $L_K\times(2K-L_K+1)$  Hankel matrices $G^{(0)}$, $G^{(1)}$, where $G^{(a)}_{i,j} = g(i+j+a-K)$ for $i \in \{0, 1, ... L_K-1\}, j \in \{0, 1, ..., 2K-L_K\}$, $a=0,1$, with $L_K = \lfloor (K+1)/2 \rfloor$, and using $g(-{\sf k})=g^*({\sf k})$.
\item Construct the $L_K\times L_K$ shift matrix $T$ by least-squares minimization of the matrix $2$-norm $\|TG^{(0)} - G^{(1)}\|$.
\item Calculate the eigenvalues of $T$, $\lambda_j = |\lambda_j|e^{i\tilde{\theta}_j}$ and from there the phase estimates $\tilde{\theta}_j$.
\item Calculate the overlap estimates $\tilde{A}_j$ by least-squares minimization of the vector $2$-norm $\|BA - g\|$, where $B$ is the $(K+1)\times L_K$ matrix
\begin{equation}
B_{k,j} = \lambda_j^{{\sf k}},
\end{equation}
and $g = [g(0), ... g(K)]^T$.
\item Return the phase estimates $\tilde{\theta}_j$ for which the corresponding overlap estimate $\tilde{A}_j\geq A$.
\end{enumerate}
\end{alg}
To use this algorithm as a subroutine in Alg.~\ref{alg:adaptive} (in place of Alg.~\ref{alg:phase_extraction}), we implement it on the matrix $V=U^{k_d}$, which requires implementing $V^{\sf k}=U^{{\sf k} k_d}$ for a range of integer ${\sf k}$ on a quantum device. 

To isolate the performance of the estimation routine from the generation of the signal itself, we do not test our protocols on data generated from simulating or approximating a particular unitary.
Instead, we test the ability of the protocols to estimate $n_\phi=2$ and $n_{\phi}=4$ randomly-chosen phases $\phi_j\in [0,2\pi]$ when sampling from the true phase function $g(k)$. We take all phases with equal weight --- $A_j = 1/n_\phi$.
We simulate the sampling from $g(k)$ in Algorithm \ref{alg:mps} or Algorithm \ref{thm:QEEP_oracle} for some $V=U^{k_d}$ by simulating the readout of a control qubit with the reduced density matrix of Eq.~\eqref{eq:control_rdm}.
(In practice this would be generated by the quantum circuit in Fig.~\ref{fig:QPE}.)
We first draw $M_d$ i.i.d. samples from the two Bernoulli distributions
\begin{align}
    \mathbb{P}_k^r(+1) = \frac{1}{2}\sum_{j=1}^{n_\phi}A_j(1+\cos(\theta_jk)),\\ \mathbb{P}_k^i(+1) = \frac{1}{2}\sum_{j=1}^{n_\phi}A_j(1-\sin(\theta_jk)),
\end{align}
where $\theta_j=k_d\phi_j\mod 2\pi$ are the eigenvalues of $V$.
Then, we return the fraction of $+1$s drawn as estimates for the real and imaginary parts of $g(k)$ respectively.
This is then repeated at all points $k=k_d{\sf k}$ for ${\sf k}=0,1,\ldots,K$.
Following the discussion in Sec.~\ref{sec:speed_limits} and using the notation from Eq.~\eqref{eq:T-upper}, we sum the total quantum cost for the algorithm over all requested $g(k)$ queries; $T=2\sum_d\sum_{{\sf k}=1}^{K}{\sf k}k_dM_d$.
(We ignore the sub-leading correction from the final term in Eq.~\eqref{eq:T-upper} as this will not affect the scaling of our result.)
For the signal length $K$ and number of points $M_d$ to sample each $g(k)$ at, we follow the bounds given in Ref.~\cite{Somma19Quantum} (both when using the QEEP and matrix pencil subroutines):
 \begin{itemize}
 \item signal length: $K = \lceil 0.1 L\ln^2 L\rceil$, with $L=\lceil\frac{2\pi}{\epsilon}\rceil$ the number of bins used in the QEEP subroutine (Def.~\ref{def:QEEP}).
 \item number of measurements of each circuit: $M_d = \left\lceil \left|\ln\left(1-p_d\right)\right| \epsilon^{-4} \right\rceil$.
 \end{itemize}
 
Here, $p_d$ is given in Eq.~\eqref{eq:defpd} for a given $k_d$.
This equation requires fixing a choice of $\alpha$ and $\gamma$ --- across all experiments we take $\alpha=2$ and $\gamma=2.1$.

\begin{figure*}
\centering
\includegraphics[width=\textwidth]{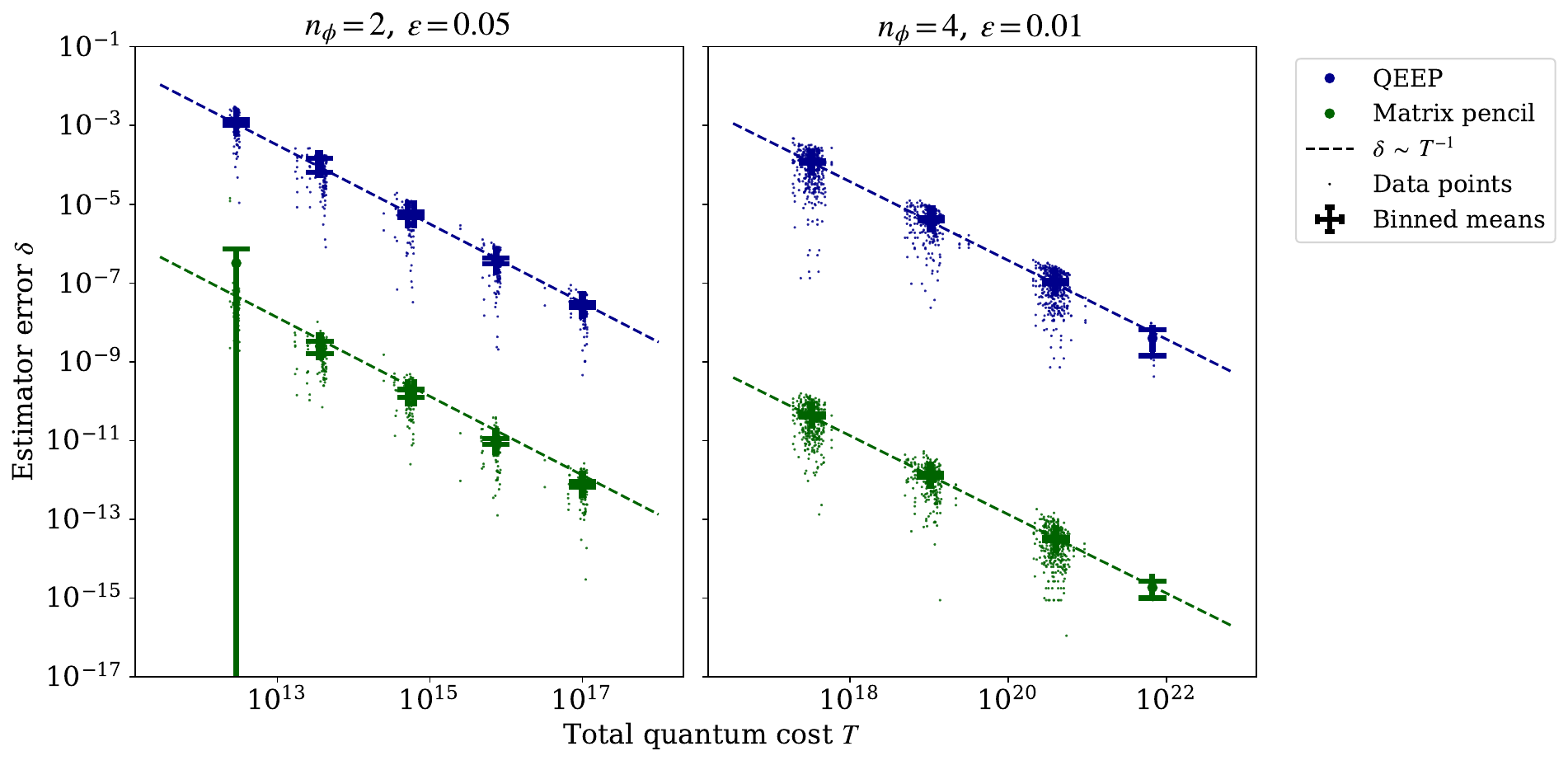}
\caption{
Convergence of the algorithm with total quantum cost $T$.
Phase estimates were obtained with either QEEP (blue) or matrix pencil (green) subroutines with parameters described in the text. 
Individual points show the error $\delta=|\tilde{\phi}_j-\phi_j|_{T}$ on individual phases in each simulation. 
This data is binned in the x-axis, and for each bin a root-mean-square error and standard deviation (error bars) are plotted in the x- and y-direction.
Dotted lines show a fit of these means to $\delta \sim T^{-1}$.
}
\label{fig:error_vs_cost}
\end{figure*}

To demonstrate that our methods achieve the Heisenberg limit, in Fig.~\ref{fig:error_vs_cost} we plot the error as a function of the total quantum cost for a set of simulations using the methods described above.
Each simulation draws a different set of $n_{\phi}$ random phases, and a final error $\delta_c\in[10^{-5},10^{-2}]$ (for each choice of $\delta_c$ we use the same $50$ sets of phases).
We plot the error for each phase estimate separately in Fig.~\ref{fig:error_vs_cost} (i.e. each simulation corresponds to $n_{\phi}$ points in the plot).
As both the total quantum cost and the error is different between simulations, we bin all experiments within a range of $T$ values, and calculate the root mean square error and root mean square total quantum cost.
This gives a good approximation to the accuracy error defined in Def.~\ref{def:MEEP} for the restricted data set used.

For the QEEP subroutine, we observe a clear fit of the data (blue points) to a $\delta\sim T^{-1}$ trend, as expected from Thm.~\ref{thm:adaptive_gets_Heisenberg_limit}, but with a rather large constant factor; we find $T\sim 10^{10}\delta^{-1}$ for estimating $2$ phases and $T\sim 10^{15}\delta^{-1}$ for estimating $4$.
Further optimization of the QEEP algorithm for these purposes may yet improve this constant factor.
However, as the methods of Ref.~\cite{Somma19Quantum} were not designed for estimating individual phases, it may be expected that this method performs somewhat badly for this purpose, so we have not pursued this further.

Simulations using the matrix pencil subroutine outperform simulations using the QEEP subroutine by a factor of $10^4-10^6$, and clearly demonstrate Heisenberg-limited scaling $\delta\sim T^{-1}$ as well.
We take this result instead of an analytic proof as strong numerical evidence for Heisenberg-limited scaling when a version of Alg.~\ref{alg:adaptive} is constructed using the matrix pencil method as a subroutine.
We notice that the error in two phase estimates in the first bin of the matrix pencil method for $n_{\phi}=2$ is significantly above the remainder of the population (by about a factor $100\times$), which blows up our error bars for this bin.
Further investigation shows that the two phases in question are from the same simulation, and separated by only $1.5\times 10^{-4}$.
By contrast, for the simulation in question (at $d=1$) our algorithm sets $k_1K\sim 3\times 10^3 < (1.5\times 10^{-4})^{-1}$ (where $k_1K$ is the largest value of the phase function $g(k)$ sampled during this simulation).
This implies that our signal lies within the region where improving our estimation accuracy by increasing the number of shots $M_d$ is exponentially hard~\cite{Moitra14Super}.
In latter simulations with these two phases we see that our estimation error regresses to similar results as all other estimates.

\section{Conclusion}
\label{sec:conclusion}

In this work we studied Heisenberg-limited quantum phase estimation using a single control qubit. In this form of phase estimation, we rely on classical signal processing to extract eigenvalue data from the phase function $g(k)$ in Eq.~\eqref{eq:phase_function}.

It has been an open question whether these methods can achieve the Heisenberg limit in the case of multiple phases: Ref.~\cite{LT:heisenberg} answered this question up to log factors with a Heisenberg-limited Monte Carlo algorithm, providing a sampling of the spectral function $A(\phi)$ in Eq.~\eqref{eq:spectral_function} from which to estimate the phases.
In this work we also answered this question in the affirmative exactly with a new adaptive multi-order phase estimation algorithm, for which we prove Heisenberg scaling if the algorithm uses a QEEP phase extraction subroutine.
We numerically show the performance of this algorithm, also when instead of using a QEEP subroutine, one uses the matrix pencil method to extract phase estimates from the phase function $g(k)$.
This result complements the previous work discussed in the introduction by closing the question of whether there exists a gap between classical post-processing of phase function data and fully quantum phase estimation.

In obtaining our results we encountered at least two details of quantum phase estimation that we have not seen discussed in the literature.
The first is the dense signal limit, Eq.~\eqref{eq:DSL} in Thm.~\ref{thm:heisenberg_limit}: sampling $g(k)$ at {\em all} integer point $k=1,\ldots,K$ is sub-optimal regardless of what method is used to process the data. However, we also briefly argued that by picking points among $k=1,\ldots,K$ at random one may go beyond this, and one could interpret this as allowing the results obtained in \cite{LT:heisenberg} in which such randomized choices for $k$ are taken.

The second is the need for adaptive choices of $k_d$ to solve the phase matching problem. It is unclear to us how far this problem extends; although Lemma~\ref{lem:phase_matching_solution} provides a practical solution, others may still exist. 
Another open question with respect to Algorithm \ref{alg:adaptive} is whether one can remove the need to choose real-valued multipliers $\kappa_d$ and restrict to integer choices. Restricting $\kappa_d\in\NN$ would significantly simplify some technical issues, i.e. the applicability of Lemma \ref{lem:normscaling} and the need for shifting phases in step \ref{step:shift} of Alg.~\ref{alg:adaptive}, but we don't know how to prove a version of Lemma~\ref{lem:phase_matching_solution} for $\kappa_d\in\NN$. In fact, we don't know whether there is a fundamental difference in performance between only using data obtained with integer $k$ in $g(k)$ versus data obtain with real-valued, --in practice rational--, $k$ in $g(k)$. 

We have assumed in our problem description, Def.~\ref{def:MEEP}, that the spectrum in the input state is discrete consisting of $n_{\phi}$ phases with amplitudes above some cut-off. In practice this condition may not be fulfilled and thus studying the performance of the algorithms on more typical spectra induced by many-body Hamiltonians and easy-to-prepare input states will be of interest. 
The scaling of our (provable) Heisenberg-limited algorithm in $n_{\phi}$ is also rather poor [$O(n_{\phi}^{24})$ as described], as we have not attempted to optimize this aspect.
This should in principle be immediately reducible to $O(n_{\phi}^3\delta^{-1}+n_{\phi}^6)$ under the assumption that the matrix pencil method continues to achieves the dense sampling limit when estimating multiple eigenvalue, however we do not know a proof of this.
In principle linear scaling with $n_{\phi}$ should be achievable (or even sub-linear if the methods of~\cite{Lin20Preparation} could be applied in this setting); optimizing this is a clear target for future work.

A direction for future research is to make this algorithm efficient in practice (i.e. improve the parameter dependence and the practical run time) or devise yet-different Heisenberg-scaling algorithms and examine their performance in the presence of experimentally-noisy signals $g(k)$.

\section*{Acknowledgements}

AD acknowledges financial support in the form of a scholarship from the Den Adel Fund.


\bibliographystyle{unsrtnat}
\bibliography{HLQPE}

\appendix

\section{Phase extraction subroutine \label{app:QEEP}}

For our Heisenberg-limited algorithm, at each order $d$ we need to extract eigenphases of $V=U^{k_d}$ from the signal generated by PFE. We require that the phase estimates satisfy the promises of Lem.~\ref{lem:phase_extraction_promises}. This is achieved by Alg.~\ref{alg:phase_extraction}, which is based on the solution to QEEP of Ref.~\cite{Somma19Quantum}.

In this appendix, we first summarize the results of Ref.~\cite{Somma19Quantum} by giving a precise definition of the QEEP (Def.~\ref{def:QEEP}) and performance guarantees of the time-series analysis (Def.~\ref{thm:QEEP_oracle}).
Then we describe the Conservative QEEP Eigenvalue Extraction algorithm (Alg.~\ref{alg:phase_extraction}) and give the proof of Lem.~\ref{lem:phase_extraction_promises}.

\subsection{The Quantum Eigenvalue Estimation Problem}

\begin{dfn}[Quantum Eigenvalue Estimation Problem, QEEP]
Let $A(\phi)$ be the spectral function defined in Eq.~\eqref{eq:spectral_function} for a unitary $U$ and $\ket{\Psi}$. Given is an error parameter $\epsilon > 0$, a confidence parameter $1 \geq p > 0$, and a set of non-negative (approximate indicator) functions $f^l(\phi)$ for $\phi\in [0,2\pi)$ for $l=0,\ldots, L-1$, $L=\lceil\frac{2\pi}{\epsilon}\rceil$, where $f^l(.)$ has support on only the interval bin
    \begin{equation}
        \Vv_l=[(l-1)\epsilon, (l+1)\epsilon]_{T},\label{eq:vl_def}
    \end{equation}
    and $f^l(\phi)+f^{(l-1)}(\phi)=1$ for all $\phi\in\Vv_l\cap\Vv_{l-1}$.
Assuming access to the {\rm PFE} subroutine, Def.~\ref{def:phase_function_estimation_problem}, the goal is to output an approximation $\tilde{b}_l$ for $l=0,\ldots,L-1$ to the integral
\begin{equation}
    b_l=\int_0^{2 \pi} d\phi A(\phi)f^l(\phi),\label{eq:sf_amplitudes_def}
\end{equation}
which satisfies 
\begin{equation}
\sum_{l=0}^{L-1} |\tilde{b}_l-b_l|\leq \epsilon,\label{eq:QEEP_bound}
\end{equation}
with probability at least $p$.
\label{def:QEEP}
\end{dfn}

Note that the bins $\Vv_l$ have width $2\epsilon$ and overlap on a region of width $\epsilon$ and 
$\sum_{l=0}^{L-1} b_l=1$.

\begin{thm}[QEEP Algorithm~\cite{Somma19Quantum}]\label{thm:QEEP_oracle}
One can solve the QEEP problem in Definition \ref{def:QEEP} with $\{f^l(.)\}$ a set of `bump' functions
\begin{align}
    f^{l}(\phi)=&\frac{2a}{\epsilon}\int_{l\epsilon-\frac{\epsilon}{2}}^{l\epsilon+\frac{\epsilon}{2}}d\phi'\nonumber\\
    &\times\exp\left\{-\left[1-\frac{4}{\epsilon^2}(\phi-\phi')^2\right]^{-1}\right\},
    \label{eq:bump}
\end{align}
with normalization constant $a\approx 2.252$, using {\rm PFE} in Def.~\ref{def:phase_function_estimation_problem} with $k=0, 1,\ldots, K$ with $K=O\left(\epsilon^{-1}\ln^2(\epsilon^{-1})\right)$ and $M=\tilde{O}(|\ln(1-p)|\epsilon^{-4})$ for each $k=1,2\ldots, K$.
The total quantum cost for $U$ is then bounded as $T =O(M K^2)=\tilde{O}(|\ln(1-p)|\epsilon^{-6})$.
\end{thm}

We note that the approximate indicator functions $f^l(.)$ in Eq.~\eqref{eq:bump} are designed to have a quickly decaying Fourier series, which is required to achieve polynomial-time scaling. We also refer the reader to Ref.~\cite{Roggero20Spectral}, which has extended the result by relaxing the requirement that $f^l(\phi)+f^{(l-1)}(\phi)=1$ on the interval $\Vv_{l}\cap\Vv_{l-1}$. The idea of this QEEP algorithm is as follows. Since $b_l=\sum_j A_j f^l(\phi_j)$, using Eq.~\eqref{eq:spectral_function}, periodically extending $f^l(\phi)$ beyond $[0,2\pi)$ and Fourier decomposing gives $b_l=\sum_j A_j \sum_{k \in \mathbb{Z}} e^{ik \phi_j} \tilde{f}^l(k)=\sum_{k \in \mathbb{Z}} g(k) \tilde{f}^l(k)$. At the same time, the fact that $f^l(.)$ is an indicator function ensures that $b_l \approx \sum_{\phi_j \in \Vv_l} A_l$. Thus knowledge of $g(k)$ and the Fourier coefficients $\tilde{f}^l(k)$ for a range of $k$ allows one to estimate the weights $b_l$.
The requirement to estimate the spectral function to within a $1$-norm $\epsilon$, Eq.~\eqref{eq:QEEP_bound}, is very stringent, hence the scaling of $T$ with error $\epsilon$ is quite costly, $T=O(\epsilon^{-6})$. 
It is possible that one can improve the scaling by re-examining the analysis in \cite{Somma19Quantum}.\\

\subsection{Conservative QEEP Eigenvalue Extraction}

\begin{alg}[Conservative QEEP Eigenvalue Extraction]\label{alg:phase_extraction}
Fix an overlap bound $A$, an error bound $0 < \epsilon < \frac{A}{3}$, and a confidence bound $0<p<1$.
Assume access to a QEEP Algorithm \ref{thm:QEEP_oracle} for a unitary $V$.
The algorithm proceeds as follows:
\begin{enumerate}
    \item Use the QEEP subroutine with error $\epsilon$, Alg.~\ref{thm:QEEP_oracle}, and confidence $p$ to generate an estimate $\tilde{b}_l$ for $b_l$ as defined in Eq.~\eqref{eq:sf_amplitudes_def}.
    \item \label{step:Sconstruct} Construct the set 
    \begin{equation}
        S=\{l\in\{0,\ldots,L-1\}| \tilde{b}_l \geq A/3\}.
    \end{equation}
    \item Find the smallest $l \in \{0,\ldots,L\}$ with $l \not \in S$ and call it $l_{\min}$.
    \item For $l'=l_{min},\ldots, (l_{min}+ L-1)\bmod{L}$:
    \begin{enumerate}
        \item if $l' \in S$ and $(l'- 1)\bmod{L}\in S$, remove $l'$ from~$S$.
    \end{enumerate}
    \item Return the set $\{\tilde{\theta}_l=l\epsilon\}_{l\in S}$ as a set of estimates of eigenphases of $V$.
\end{enumerate}
\end{alg}

Let us now motivate Algorithm \ref{alg:phase_extraction}. One may identify phases with sufficient probability in the output of the QEEP algorithm as the bins $l$ with $b_l > b_{\mathrm{cutoff}}$ with $b_l$ in Eq.~\eqref{eq:sf_amplitudes_def}.
Then to convert this into an estimate of a phase $\theta_j$, for each such bin above cut-off we could output the estimate $\epsilon l$, i.e. in the middle of the corresponding bin $\Vv_l$.

We calculate appropriate values of $b_{\mathrm{cutoff}}$ and the QEEP error $\epsilon$ to guarantee that we output an estimate for each $\theta_j$ with $A_j>A$ with a provable confidence, and to guarantee no estimate in the absence of any $\theta_j$.
Def.~\ref{def:QEEP} states that when there exists such a $\theta_j\in\Vv_{l}\cap\Vv_{l-1}$, we are guaranteed with confidence $p$ that $\tilde{b}_l+\tilde{b}_{l-1}+\epsilon > A$ (as $b_l+b_{l-1} > A$). To guarantee that at least one of $\tilde{b}_l$ or $\tilde{b}_{l-1}$ is larger than $b_{\mathrm{cutoff}}$ with the same confidence, we thus require $b_{\mathrm{cutoff}} \leq (A - \epsilon) / 2$.

Similarly, Def.~\ref{def:QEEP} states that when there exists no $\theta_j\in\Vv_{l}$ with $A_j>0$, we are guaranteed with confidence $p$ that $b_l<\epsilon$.
To prevent a spurious estimate in this case, we require $b_{\mathrm{cutoff}}\geq\epsilon$.
Solving this to maximise $\epsilon$ (which minimizes the cost of the QEEP routine) yields $b_{\mathrm{cutoff}}=\epsilon=A/3$ which is what we use in Alg.~\ref{alg:phase_extraction}.

A further small complication exists in solving the problem posed in Def.~\ref{def:MEEP}: we require that one outputs at most $n_{\phi}$ phases.
This will be satisfied if we can guarantee at most one phase estimate per $\theta_j$ with $A_j>A$, as we know there are $n_{\theta}=n_{\phi}$ such estimates.
As the bins $\Vv_l$ (Def.~\eqref{eq:vl_def}) overlap, a phase $\theta_j$ may participate in up to two neighbouring bins --- corresponding to amplitudes $b_l,b_{l+1}>A/3$.
To ensure that this does not result in two estimates being generated, one could take all contiguous sets $b_{l_1},b_{l_1+1},\ldots,b_{l_2}>A/3$ and prune away every second index $l$.
In doing this we need to respect the periodicity of the $\Vv_l$: $\Vv_{L-1}\cap\Vv_0\neq\emptyset$, so these contiguous sets may wrap around the circle. Pruning every second index when starting from the middle of one of these contiguous sets may result in two neighbouring $l, l+1$  being removed, which is not what we desire.
Instead, in the following pseudocode, after generating the set of all $l$ with sufficient $b_l$, we find the first gap (in $l$) between these regions (corresponding to the first $b_l<A/3$).
We then iterate (from this point $l_{\min}$ to $L-1$ and then from $0$ to $l_{\min}$) over the $b_l$, and remove each $l$ from our set if $b_{l-1}>A/3$ and $l-1$ was not itself removed.

\subsection{Proof of Lemma \ref{lem:phase_extraction_promises}}

Lemma \ref{lem:phase_extraction_promises} relates the performance of Alg.~\ref{alg:phase_extraction} to the performance of the QEEP subroutine.
We now prove this Lemma.
Note that $\{\theta_j\}$ is an ordered list of $n_{\phi}$ phases, while $\{\tilde{\theta}_l\}$ is an ordered list (that we will show contains at most $n_{\phi}$ phases). 

{\em Proof of Lemma \ref{lem:phase_extraction_promises}}
Our proof follows by showing that the output from Alg.~\ref{alg:phase_extraction} satisfies these statements whenever Eq.~\eqref{eq:QEEP_bound} holds.
This yields our confidence bound as Eq.~\eqref{eq:QEEP_bound} holds with probability $p$ by Def.~\ref{def:QEEP}.

To see that statement \ref{stat2} holds when Eq.~\eqref{eq:QEEP_bound} is satisfied, note that if $\Vv_l$ contains no phases $\theta_j$ with $A_j>A$ and Eq.~\eqref{eq:QEEP_bound} is satisfied, $\tilde{b}_l<\epsilon<A/3$, and $l$ will not be added to the set $S$ in Alg.~\ref{alg:phase_extraction}.
This implies that when Eq.~\eqref{eq:QEEP_bound} holds, if $l\in S$ there exists some $\theta_j\in\Vv_l$ with $A_j>A$, in which case $|\theta_j-\tilde{\theta}_l|_T=|\theta_j-\epsilon l|_T\leq 2\epsilon$.

To see that statement \ref{stat3} holds when Eq.~\eqref{eq:QEEP_bound} is satisfied, note that $\theta_j\in \Vv_l\cap\Vv_{(l+1)\bmod{L}}$ for exactly one $l$ (and $\theta_j\notin\Vv_{m}$ for $m\neq l,(l+1)\bmod{L}$).
Then, when Eq.~\eqref{eq:QEEP_bound} holds, $\tilde{b}_l+\tilde{b}_{(l+1)\bmod{L}}>2A/3$, so $\max(\tilde{b}_l, \tilde{b}_{(l+1)\bmod{L}})>A/3$, and either $l$, $(l+1)\bmod{L}$ or $l$ and $(l+1)\bmod{L}$ will be added to the set $S$ during step \ref{step:Sconstruct} of Alg.~\ref{alg:phase_extraction} for each phase $\theta_j$.
Then, step 4 of Alg.~\ref{alg:phase_extraction} will remove $(l+1)\bmod{L}$ if $l$ remains in $S$, so each phase $\theta_j$ can contribute to only one final estimate $\tilde{\theta}_l=l\epsilon$, and the number of estimates is bounded from above by the number of phases.

To see that statement \ref{stat1} holds, we use the point in the previous paragraph that, when Eq.~\eqref{eq:QEEP_bound} is satisfied, each phase $\theta_j$ with $A_j>A$ adds either $l$, $(l+1)\bmod{L}$, or $l$ and $(l+1)\bmod{L}$ to the set $S$ during step 2. of Alg.~\ref{alg:phase_extraction}.
Then in step 4. of Alg.~\ref{alg:phase_extraction}, $l$ is removed from $S$ only if $(l-1)\bmod{L}$ remains in $S$, and $(l+1)\bmod{L}$ is removed from $S$ only if $l$ remains in $S$.
This implies that the distance from $\theta_j$ to an estimate is bounded by
\begin{multline}
    \max_{\theta_j\in\Vv_l\cap\Vv_{(l+1)}}\hspace{0.2cm}\max_{l'=l-1, l,l+1 \mod{L}}|\theta_j-\epsilon l'|_T =\\= \epsilon (l+1)\bmod{L} -\epsilon(l-1))\bmod{L}=2\epsilon,
\end{multline}
as required. \qedsymbol
\section{Proof of Lemma \ref{lem:phase_matching_solution}}\label{app:phase_matching_proof}

Let us first prove the existence of $\kappa_{d+1} \in [2,\kappa_{\rm max}]$ with $\kappa_{\rm max}=3$ that satisfies our conditions (the proof for $k_1$ is similar), for small enough $\epsilon$ in Eq.~\eqref{eq:eps-n}. Note that Eq.~\eqref{eq:eps-n} implies 
\begin{equation}
  \epsilon \leq \frac{\pi}{300} \approx 0.01.
    \label{eq:epsbound}
\end{equation}
Given some pair $\tilde{\phi}^{(d)}_j \neq \tilde{\phi}^{(d)}_l$, $j <  l$, let $\Delta_{j,l} = |\tilde{\phi}^{(d)}_j-\tilde{\phi}^{(d)}_l|$ and let $R_{j,l}$ be a set of $\kappa_{d+1}$ such that  neither Eq.~\eqref{eq:aliasing_lower_bound} nor Eq.~\eqref{eq:aliasing_bound} holds for the chosen phases, that is,
\begin{align}
    R_{j,l} = \bigg\{&
    \kappa_{d+1} \in [2, \kappa_{\max}]:\nonumber\\
    &|\Delta_{j,l}|_T \geq \frac{\pi-2\epsilon(1+\kappa_{d+1})}{k_d \kappa_{d+1}} \nonumber\\
    \land
    &|k_d\kappa_{d+1}\Delta_{j,l}|_T \leq 4\epsilon(1+\kappa_{d+1})
    \bigg\}.
\end{align}

We call $\cup_{j,l}R_{j,l}$ the forbidden region and want to show that we can choose a value for $\kappa_{d+1} \in [2,3]$ outside this forbidden region if $\epsilon$ is sufficiently small. We do this by bounding the size of the forbidden region above and showing that this is smaller than the region $[2,3]$, leaving room to choose $\kappa_{d+1}$.

Note that $R_{j,l}$ is nonempty only if
\begin{align}
    k_d \Delta_{j,l} &\geq k_d |\Delta_{j,l}|_T \notag \\
    &\geq 
        \frac{\pi-2\epsilon(1+\kappa_{d+1})}{\kappa_{d+1}} \notag \\
        &\geq 
    \frac{\pi-2\epsilon(1+\kappa_{\max})}{\kappa_{\rm max}} \notag \\
    &> \frac{73\pi}{225},
\label{eq:appb_Delta_bound}
\end{align}
using Eq.~\eqref{eq:epsbound}.

We may write the set $R_{j,l}$ as
\begin{align}
    R_{j,l}=&\left[\max\left(2,\frac{\pi-2\epsilon}{k_d|\Delta_{j,l}|_T+2\epsilon}\right),\kappa_{\max}\right]\nonumber\\&\cap\bigcup_{n \in \NN} R_{j,l}^{(n)},
\end{align}
where $R_{j,l}^{(n)}$ is the set of $\kappa_{d+1}$ for which 
\begin{align}
    |k_d\kappa_{d+1}\Delta_{j,l}|_T &=\bigg| k_d\kappa_{d+1}\Delta_{j,l} - 2\pi n\bigg| \notag \\
    &\leq 4\epsilon (1 + \kappa_{d+1}),
      \label{def:rjln}
\end{align}
for some $n\in \mathbb{N}$.

Solving this equation for $\kappa_{d+1}$ yields
\begin{align}
    R_{j,l}^{(n)} = &\left[\frac{2\pi n - 4\epsilon}{k_d\Delta_{j,l} + 4\epsilon},\frac{2\pi n + 4\epsilon}{k_d\Delta_{j,l} - 4\epsilon}\right].
\end{align}
The size of each interval $R_{j,l}^{(n)}$ can then be calculated
\begin{equation}
    \left|R_{j,l}^{(n)}\right| = \frac{8\epsilon(2n\pi+ k_{d}\Delta_{j,l})}{k_{d}^2\Delta_{j,l}^2-16\epsilon^2}.
\end{equation}
We can bound
\begin{align}
    \left|R_{j,l}\right|&\leq \left|\bigcup_{n=1}^{n_{\max}}R^{(n)}_{j,l}\right|
    \leq \sum_{n=1}^{n_{\max}}\left|R^{(n)}_{j,l}\right| \notag \\
    &=\frac{8\epsilon}{k_{d}^2\Delta_{j,l}^2-16\epsilon^2}
    \bigg(\pi n_{\max}(n_{\max}+1)\nonumber\\
    &\qquad\qquad\qquad+k_d\Delta_{j,l}n_{\max}\bigg) .
\end{align}
Here, $n_{\max}=n_{\rm max}(j,l)$ is the largest index of a set $R^{(n)}_{j,l}$ in Eq.~\eqref{def:rjln} for which $\kappa_{d+1} \in [2,\kappa_{\max}]$. Since $\kappa_{d+1} \leq 3$, Eq.~\eqref{def:rjln} implies that 
\begin{equation}
    n_{\max}\leq \frac{3(k_d \Delta_{j,l} + 4\epsilon)+4\epsilon}{2\pi},
\end{equation}
Now using Eq.~\eqref{eq:appb_Delta_bound} and Eq.~\eqref{eq:epsbound} gives 
   \begin{align}
       \big|&R_{j,l}\big|
       \leq \frac{1}{1- (4 \epsilon/k_d \Delta_{j,l})^2}\nonumber \notag \\
       &\times\frac{30 \epsilon (864000 \epsilon^2+ 248160 \epsilon \pi+11899 \pi^2)}{5329 \pi^3} \notag \\
       &\leq 
       \frac{3 \epsilon (864000 \epsilon^2+ 248160 \epsilon \pi+11899 \pi^2)}{532 \pi^3} \notag \\
       &\leq 0.24,
       \label{eq:upperR}
   \end{align}
   where the last inequality used Eq.~\eqref{eq:epsbound}.
   As there are $n_{\phi}\geq 2$ phases the length of the total forbidden region $\cup_{j,l}R_{j,l}$ is bounded from above by $\frac{n_{\phi}^2}{2} \left|R_{j,l}\right|$. We want to this interval to be, say, at most $1/4$, so that by choosing $\kappa_{d+1}$ randomly we have a 75$\%$ change of not landing in the forbidden interval. For larger $n_{\phi}$ we thus should use
\begin{equation}
    \epsilon\leq \epsilon_{\mathrm{crit}} = \frac{2\pi}{300 n_{\phi}^2},
\end{equation}
leading to Eq.~\eqref{eq:eps-n}.

We now repeat the above approach for the special case of finding the multiplier in the first round $\kappa_1=k_1$. Consider thus $k_1\in [3n_{\phi},\kappa_{\max}]$ with $\kappa_{\rm max}=3 n_{\phi}+1$.
The key difference here is that there is a stricter lower bound on this multiplier $\kappa_{\rm max}$ so that $\epsilon$ needs to be chosen smaller, depending on $n_{\phi}$, namely we choose 
\begin{equation}\epsilon_0 \leq \epsilon_{\mathrm{crit},0} = \frac{2\pi}{300 n_{\phi}^4},
\label{eq:eps-choice}
\end{equation}
as expressed in Eq.~\eqref{eq:eps-n0}.
   
Given some pair $\tilde{\phi}^{(0)}_j \neq \tilde{\phi}^{(0)}_l$, $j <  l$, and again let $\Delta_{j,l} = |\tilde{\phi}^{(0)}_j-\tilde{\phi}^{(0)}_l|$.
Then, let $R_{j,l}$ be the set of $k_1$ such that neither Eq.~\eqref{eq:aliasing_lower_bound} nor Eq.~\eqref{eq:aliasing_bound} holds for the chosen phases.
That is,
\begin{align}
    R_{j,l} = \bigg\{&
    k_1 \in \left[3n_\phi, \kappa_{\max}\right]:\nonumber\\
    &|\Delta_{j,l}|_T \geq \frac{\pi-2\epsilon_0(1+k_1)}{k_1}\nonumber\\
    \land&
    |k_1\Delta_{j,l}|_T \leq 4\epsilon_0(1+k_1)
    \bigg\}.
\end{align}

We again call $\cup_{j,l}R_{j,l}$ the forbidden region and want to show that we can choose a value for $k_1 \in [3n_\phi,3 n_{\phi}+1]$ outside this forbidden region, assuming that $\epsilon_0$ is small enough. Note that the logic of the first few inequalities in Eq.~\eqref{eq:appb_Delta_bound} still holds in this new calculation, leading to
\begin{align}
    \Delta_{j,l}&\geq  \frac{\pi-2 \epsilon_0(1+\kappa_{\rm max})}{\kappa_{\rm max}} \notag \\
    &\hspace{1cm}\geq \frac{\pi(1-\frac{2+3n_{\phi}}{75n_{
    \phi}^4})}{1+3n_{\phi}},
    \label{eq:delbound}
\end{align}
where the second inequality used Eq.~\eqref{eq:eps-choice} and the value for $\kappa_{\rm max}$. Note that for large $n_{\phi}$ this allows $\Delta_{j,l}$ to decrease like $\sim 1/n_{\phi}$, while previously $\Delta_{j,l}$ was lowerbounded by a constant, Eq.~\eqref{eq:appb_Delta_bound}.

This time, we may write the set $R_{j,l}$ as
\begin{align}
    R_{j,l}=&\left[\max\left(3n_\phi,\frac{\pi-2\epsilon_0}{|\Delta_{j,l}|_T+2\epsilon_0}\right),\kappa_{\max}\right]\nonumber\\&\cap\bigcup_{n \in \NN} R_{j,l}^{(n)},
\end{align}
where $R_{j,l}^{(n)}$ is the set of $k_1$ satisfying
\begin{equation}
    \bigg| k_1\Delta_{j,l} - 2\pi n\bigg| \leq 4\epsilon_0 (1+k_1),
      \label{def:rjln-2}
\end{equation}
for some $n\in \mathbb{N}$.
Solving this equation for $k_1$ yields
\begin{align}
    R_{j,l}^{(n)} = &\left[\frac{2\pi n - 4\epsilon_0}{\Delta_{j,l} + 4\epsilon_0},\frac{2\pi n + 4\epsilon_0}{\Delta_{j,l} - 4\epsilon_0}\right].
\end{align}
with length
\begin{equation}
    \left|R_{j,l}^{(n)}\right| = \frac{8\epsilon_0(2n\pi+ \Delta_{j,l})}{\Delta_{j,l}^2-16\epsilon_0^2}.
\end{equation}
We can then bound
\begin{align}
    \left|R_{j,l}\right|&\leq \left|\bigcup_{n=n_{\min}}^{n_{\max}}R^{(n)}_{j,l}\right| \leq \sum_{n=n_{\min}}^{n_{\max}}\left|R^{(n)}_{j,l}\right| \notag \\
    &=\frac{8\epsilon_0}{\Delta_{j,l}^2-16\epsilon_0^2}\bigg[\pi (n_{\max}^2-n_{\min}^2)\nonumber\\
    &\qquad+
    (\Delta_{j,l}+\pi)
    (n_{\max}-n_{\min})\nonumber\\
    &\qquad+
    2\pi n_{\min}+\Delta_{j,l}\bigg].
\label{eq:R-upper}
\end{align}
Here, $n_{\max}=n_{\rm max}(j,l)$ and $n_{\min}=n_{\min}(j,l)$ are the largest and smallest indices of sets $R^{(n)}_{j,l}$ in Eq.~\eqref{def:rjln-2} for which $k_{1} \in [3n_{\phi},3 n_{\phi}+1]$. Finding the minimal and the maximal value for $n$ in Eq.~\eqref{def:rjln-2} given the bounds on $k_1$ gives
\begin{multline}
  \frac{3n_{\phi}(\Delta_{j,l} - 4\epsilon_0)-4\epsilon_0}{2\pi} \leq n_{\rm min} \leq n_{\max}\\
  \leq \frac{(3n_{\phi}+1)(\Delta_{j,l} + 4\epsilon_0)+4\epsilon_0}{2\pi}.
\end{multline}
As there are $n_{\phi}$ phases, the length of the total forbidden region $\cup_{j,l}R_{j,l}$ is bounded from above by $\frac{n_{\phi}^2}{2} \left|R_{j,l}\right|$.  By plugging the bound for $\epsilon_0$ in Eq.~\eqref{eq:eps-choice} and the bound for $\Delta_{j,l}$ in Eq.~\eqref{eq:delbound} together into Eq.~\eqref{eq:R-upper}, one can verify that
\begin{widetext}
{\scriptsize
\begin{equation}
    |R_{j,l}| \frac{n_{\phi}^2}{2} \leq \frac{-16-168 n_{\phi}-621 n_{\phi}^2+552 n_{\phi}^4 + 14400 n_{\phi}^5 + 44550 n_{\phi}^6 + 40500 n_{\phi}^7 + 73125 n_{\phi}^8 + 
 270000 n_{\phi}^9 + 202500 n_{\phi}^{10}}{450 n_{\phi}^3(-4-9n_{\phi}-100 n_{\phi}^3-150 n_{\phi}^4+1875 n_{\phi}^7)},
\end{equation}}
\end{widetext}
which can be verified to be less than 0.5 for all $n_{\phi}$. Hence a random choice for $k_1$ in the interval $[3 n_{\phi},3 n_{\phi}+1]$ gives at least a 50\% chance to not land in the forbidden region.
\qed

\section{Range of shifted phase estimates}
\label{app:shift}

In this Appendix we prove that when the unitary $U$ is shifted as in step \ref{step:shift} of Alg.~\ref{alg:adaptive}, and as long as the output of the phase extraction subroutine (Alg.~\ref{alg:phase_extraction}) meets the promises given in Lem.~\ref{lem:phase_extraction_promises}, all phase estimates of $U^{k_d}$ will lie in the region for which Lem.~\ref{lem:normscaling} holds.
This allows us to invoke Lem.~\ref{lem:normscaling} as required during Lem.~\ref{lem:phase_matching_checks} and Lem.~\ref{lem:phase_matching_with_failures}. Note that if we were to shift the spectrum such that the middle of the largest gap would sit at 0, we would do $U \rightarrow U e^{-i \zeta}$. However, the `stay-away-from-the-boundary' condition of Lem. \ref{lem:normscaling} is not symmetric, hence we shift by a different amount which also depends on the error $\epsilon_0$ in the phase estimates. 

\begin{lem}
\label{lem:shifting_unitary}
Let $\{\phi_j\}$ be the list of eigenphases of unitary $U$, and let $n_\phi$, $\{\tilde\phi_l^{(0)}\}$, $\zeta$, $d_\zeta$ be as defined in steps \ref{step:order0} and \ref{step:shift} of Alg. \ref{alg:adaptive}. Assume:
\begin{itemize}
    \item (Assumption 1a) For every phase $\phi_j$, there exists an estimate $\tilde{\phi}_l^{(0)}$ such that $|\phi_j-\tilde{\phi}_l^{(0)}|_T\leq 2\epsilon_0$.
    \item (Assumption 1b) For every estimate $\tilde{\phi}_l^{(0)}$, there exists a phase $\phi_j$ such that $|\phi_j-\tilde{\phi}_l^{(0)}|_T \leq 2\epsilon_0$.
    \end{itemize}
Then for all $k \geq 3n_\phi$, for the estimated shifted eigenphases $\{\tilde{\varphi}_l^{(0)}\}$ it holds that
\begin{equation}
     \frac{\pi}{k}+ 16\epsilon_0\leq  \tilde{\varphi}_l^{(0)} \leq \frac{\pi(2\lfloor k\rfloor-1)}{k}-16\epsilon_0.
     \label{eq:estim}
\end{equation}
which implies Eq.~\eqref{eq:normscaling} for $\tilde{\varphi}_l^{(0)}=\phi$. In addition, the eigenphases $\{\varphi_j\}$ of the shifted unitary $Ue^{-i(\zeta+d_\zeta/2-8\epsilon_0)}$ satisfy
\begin{equation}
    \frac{\pi}{k}+14\epsilon_0\leq \varphi_j \leq \frac{\pi(2\lfloor k\rfloor-1)}{k}-14\epsilon_0.
    \label{eq:realphases}
\end{equation}
which again implies Eq.~\eqref{eq:normscaling} for $\varphi_j=\phi$. 
\end{lem}
\begin{proof}
Let 
\begin{equation}
    \tilde\varphi^{(0)}_l = \left(\tilde\phi_l^{(0)}-\zeta-\frac{d_\zeta}{2}+8\epsilon_0\right)\textrm{ mod }2\pi.
\end{equation}
Let us first show that
\begin{equation}
\label{eq:shifted_estimates_bound}
    \tilde\varphi^{(0)}_l \in \left[\frac{\pi}{2n_\phi}+8\epsilon_0, 2\pi-\frac{3\pi}{2n_\phi}+8\epsilon_0\right].
\end{equation}
By definition of $\zeta$ (as the midway point in the largest gap) and $d_{\zeta}$ (as half the largest gap) we have
{\small
\begin{equation}
    \left(\tilde\phi_l^{(0)}-\zeta-\frac{d_\zeta}{2}\right)\textrm{ mod }2\pi \in \left[\frac{d_\zeta}{2}, 2\pi-\frac{3d_\zeta}{2}\right].
\end{equation}}
We have $d_\zeta \geq \pi/n_\phi$ (with equality corresponding to $n_\phi$ uniformly distributed estimates). By Eq.~\eqref{eq:eps-n0} it follows that 
\begin{equation}
\label{eq:eps-other-bound}
    \epsilon_0 \leq \frac{\pi}{48n_\phi},
\end{equation} and thus $\epsilon_0 < \frac{d_\zeta}{16}$, leading to Eq.~\eqref{eq:shifted_estimates_bound}. 
By the assumptions the shifted phases $\varphi_j$ lie within $2\epsilon_0$ from the estimates $\tilde\varphi^{(0)}_l$.
Thus for each $\varphi_j$ there exists a $\tilde\varphi^{(0)}_l$ such that
\begin{align}
    \varphi_j - 14\epsilon_0 &\geq \tilde\varphi^{(0)}_l - 16\epsilon_0 \notag \\&
    \geq \frac{\pi}{2n_\phi} - 8\epsilon_0 \notag \\
    &\geq \frac{\pi}{2n_\phi}-\frac{8\pi}{48n_\phi} \notag \\
    &= \frac{\pi}{3n_\phi}
    \geq \frac{\pi}{k}.
\end{align}
and 
\begin{align}
    \varphi_j + 14\epsilon_0 &\leq \tilde\varphi^{(0)}_l + 16\epsilon_0 \notag \\
    &\leq 2\pi-3\left(\frac{\pi}{2n_\phi} - 8\epsilon_0\right)\notag  \\
    &\leq 2\pi - \frac{3\pi}{k} \notag \\
    &\leq \frac{\pi(2\lfloor k\rfloor-1)}{k}.
\end{align}
where we have used Eq.~\eqref{eq:shifted_estimates_bound}, Eq.~\eqref{eq:eps-other-bound}, $k \geq 3 n_{\phi}$ and $\lfloor k \rfloor > k-1$. This implies Eq.~\eqref{eq:realphases} and also Eq.~\eqref{eq:estim}.
\end{proof}

\end{document}